\documentclass[11pt]{article} 

\usepackage{amsmath,amsthm,latexsym,amssymb,amsfonts,epsfig}

\usepackage{float}
\usepackage{cite}
\usepackage{hyperref}
\hypersetup{
    colorlinks,%
    citecolor=blue,%
    filecolor=blue,%
    linkcolor=blue,%
    urlcolor=blue
}

\newtheorem*{Lemma}{Lemma}

\newcommand{\be}{\begin{equation}}
\newcommand{\ee}{\end{equation}}
\newcommand{\ba}{\begin{eqnarray}}
\newcommand{\ea}{\end{eqnarray}}

\title{{\sf Hamiltonian Renormalisation IV.\\}
{\sf Renormalisation Flow of D+1 dimensional free scalar fields}\\
{\sf and Rotation Invariance}}
\author{
{\sf T. Lang}$^1$\thanks{{\sf 
thorsten.lang@gravity.fau.de}},
{\sf K. Liegener}$^1$\thanks{{\sf 
klaus.liegener@gravity.fau.de}},
{\sf T. Thiemann}$^1$\thanks{{\sf 
thomas.thiemann@gravity.fau.de}}\\
\\
{\sf $^1$ Inst. for Quantum Gravity, FAU Erlangen -- N\"urnberg,}\\
{\sf Staudtstr. 7, 91058 Erlangen, Germany}\\
}
\date{{\small\sf \today}}

\makeatletter
\@addtoreset{equation}{section}
\makeatother

\begin{document}

\maketitle
{\sf
\begin{abstract}
\fontfamily{lmss}\selectfont{ In this article we extend the test of Hamiltonian Renormalisation proposed in this series of articles to the $D$-dimensional case using a massive free scalar field. The concepts we introduce are explicitly computed for the $D=2$ case but transfer immediately to higher dimensions. In this article we define and verify a criterion that monitors, at finite resolution defined by a cubic lattice, 
whether the flow approaches a  {\it rotationally invariant} fixed point.}
\end{abstract}

\newpage
\tableofcontents

\newpage

~\\
{\bf Notation:}\\
\\
In this paper we will deal with quantum fields in the presence of an
infrared cut-off $R$ and with smearing functions of finite time support
in $[-T,T]$. The spatial ultraviolet cut-off is denoted by $M$ and has the
the interpretation of the number of lattice vertices in each spatial
direction. We will mostly not be interested in an analogous temporal
ultraviolet cut-off $N$ but sometimes refer to it for comparison
with other approaches. These quantities allow us to define dimensionful
cut-offs $\epsilon_{RM}=\frac{R}{M},\; \delta_{TN}=\frac{T}{N}$. In
Fourier space we define analogously
$k_R=\frac{2\pi}{R},\; k_M=\frac{2\pi}{M},\;
k_T=\frac{2\pi}{T},\; k_N=\frac{2\pi}{N}$.

We will deal with both instantaneous fields, smearing functions and Weyl
elements as well as corresponding temporally dependent objects.
The instantaneous objects are denoted by lower case letters
$\phi_{RM}, \; f_{RM},\; w_{RM}[f_{RM}]$, the temporally dependent ones by
upper case ones
$\Phi_{RM}, \; F_{RM},\; W_{RM}[F_{RM}]$. As we will see, smearing functions
$F_{RM}$ with
compact and discrete (sharp) time support will play a more fundamental role
for our purposes than those with a smoother dependence.

Osterwalder-Schrader reconstruction concerns the interplay between
time translation invariant, time reflection invariant and reflection positive
measures (OS measures)
$\mu_{RM}$ on the space of history fields $\Phi_{RM}$ and their corresponding
Osterwalder-Schrader (OS) data ${\cal H}_RM, \Omega_{RM}, H_{RM}$ where
${\cal H}_{RM}$ is a Hilbert space with cyclic (vacuum) vector $\Omega_{RM}$
annihilated by a self-adjoint Hamiltonian $H_{RM}$. Together, the vector
$\Omega_{RM}$ and the scalar product $<.,.>_{{\cal H}_{RM}}$ define
a measure $\nu_{RM}$ on the space of instantaneous fields $\phi_{RM}$.

Renormalisation consists in defining a flow or sequence
$n\to \mu^{(n)}_{RM},\; n\in \mathbb{N}_0$ for all
$M$ of families of measures $\{\mu^{(n)}_{RM}\}_{M\in \mathbb{N}}$.
The flow will be defined in terms of a coarse graining or embedding
map $I_{RM\to M'},\; M<M'$ acting on the smearing functions and satisfying
certain properties that will grant that 1. the resulting fixed point family
of measures, if it exists, is cylindrically consistent and 2. the flow
stays within the class of OS measures. Fixed point quantities are denoted
by an upper case $^\ast$, e.g. $\mu^\ast_{RM}$.

\newpage

\section{Introduction}
\label{s1}
\numberwithin{equation}{section}

To construct an example of an interacting Quantum Field Theory (QFT) on Minkowski space
satisfying the Wightman 
axioms remains a major challenge of fundamental physics. While a lot of progress has been made in the constructive programme in lower dimensions (see e.g. \cite{GJ87,Fr78,Riv00,Sim74}), still  no mathematically well-defined interacting QFT in $D+1=4$ dimensions has been found to date.  
The importance of this problem can be measured by the fact that the Clay Mathematical Institute\footnote{http://www.claymath.org/millennium-problems/yang-mills-and-mass-gap} devoted one of its millennium prizes to this research field. Due to Haag's theorem, which roughly says that 
the interacting theory cannot be defined on the same Hilbert space as the free theory, the only chance to solve this problem is to use a non-perturbative approach. One of these 
is the well-established  Lattice Quantumchromodynamics (LQCD) approach. Here one uses 
a spacetime lattice with spacing $\epsilon$ as a UV regulator that label a whole family 
of theories labelled by $\epsilon$ which are supposed to describe the theory at resolution $\epsilon$. The naive driscretisations of actions or Hamiltonians that are motivated by the classical 
theory do not define a consistent  family of theories which must be such that the measurements 
of observables at scale $\epsilon$ must give identical results no matter which theory 
of scale $\epsilon'<\epsilon$ is used. To construct such a consistent family of theories, which 
in the constructive setting is defined by a family of measures $\epsilon\mapsto \mu_\epsilon$,
one uses the idea of the renormalisation group (RG)
\cite{WK73,BW74,Wil75,Kad77,Cre83}.  
Renormalisation now consists in constructing a sequence $\mu^{(n)}_{\epsilon}$ of measure families where one obtains $\mu^{(n+1)}_{\epsilon}$ from $\mu^{(n)}_{\epsilon/2}$ by integrating out the degrees of freedom at scale $\epsilon/2$ that do not contribute to scale $\epsilon$
and where $\mu^{(0)}_\epsilon$ corresponds to the initial, naive discretisation.  If the sequence converges or at least has a fixed point (accumulation point) family $\mu^*_{\epsilon}$ then by construction that family is consistent, and chances are that it qualifies as the set of cylindrical projections 
of a corresponding physical continuum measure $\mu^*$. A lot of properties  of these fixed point theories or {\it perfect lattice measure families} $\mu^*_{\epsilon}$ haven been investigated \cite{Has98,Has08,HN93}, among them how a family of theories labelled by discrete cubic lattices
can still encode properties of the continuum such as {\it Euclidian invariance}.

In order to attack these problems from a new angle in our companion paper 
\cite{LLT1} a Hamiltonian Renormalisation formalism has been introduced. The motivation comes from the observation 
that it is much easier to compute the matrix elements of a Hamiltonian operator $H$ than 
of its contraction semigroup (Gibbs exponential $e^{-\beta H}$ at inverse temperature $\beta$) 
which is defined by a measure $\mu$ by Osterwalder-Schrader (OS) reconstruction \cite{OS72} 
provided that $\mu$ is reflection positive, among other properties. The original idea was therefore
to monitor the Wilsonian renormalisation flow of measures $n\mapsto \mu^{(n)}_\epsilon$ 
sketched above
in terms of its corresponding Osterwalder-Schrader data (OS data) of triples
 $(\mathcal{H}_{\epsilon}^{(n)},H^{(n)}_{\epsilon},\Omega^{(n)}_{\epsilon})$ consisting of a Hilbert space, a self-adjoint and positive semi-definite Hamiltonian operator thereon and a vacuum vector
which is a zero eigenvector for the Hamiltonian. Given such a renormalisation flow of OS data 
the plan was to extract a direct renormalisation flow that relates the OS data at scale $\epsilon$ 
to those at scale $\epsilon/2$ without recourse to the measure. It turns out that this idea fails 
in the sense that this measure derived flow of OS data makes it necessary to go back and forth between the OS data 
and OS measure so that one does need the matrix elements of the contraction semigroup which 
we wanted to avoid. However, the measure derived flow suggests a different, direct 
Hamiltonian  flow which does avoid the recourse to the measure.
 
In our companion paper \cite{LLT2} we checked that the proposed direct Hamiltonian flow,
while very different from the measure derived one, still defines the same continuum OS data 
as the continuum OS measure at the respective fixed points, at least for the two-dimensional,
massive Klein-Gordon model. In other words, the measure flow has a continuum fixed point 
measure $\mu^\ast$ whose OS data $({\cal H}^\ast, H^\ast, \Omega^\ast)$ coincide with the 
the continuum fixed point of the direct Hamiltonian flow which we take as an encouraging sign.
In fact, the OS data of the cylindrical projections $\mu^\ast_\epsilon$ of $\mu^\ast$ 
have OS data which have nothing to do with the family member $({\cal H}^{\ast}_\epsilon, 
H^\ast_\epsilon, \Omega^\ast_\epsilon)$ of the fixed point family of the direct Hamiltonian flow.
The latter are the physically relevant quantities since there are natural isometric 
injections $J^\ast_\epsilon {\cal H}^\ast_\epsilon\to {\cal H}^{\ast}$ which are the result of an 
inductive limit construction of Hilbert spaces such that 
$H^\ast_\epsilon=(J^\ast_\epsilon)^\dagger$ and 
$J^\ast_\epsilon \Omega^\ast_\epsilon=\Omega^\ast$. Further properties of this direct 
Hamiltonian flow touching subjects such as stability, criticality and universality are examined in 
our companion paper \cite{LLT3}.       

In this paper we are going extend \cite{LLT2,LLT3} by removing the restriction to two dimensions 
and considering the massive Klein Gordon model in all dimensions. This enables us to 
ask and answer the question, how  the finite resolution OS data, while defined on a 
cubic lattice labelled by $\epsilon$, still reveal properties that the continuum theory does have,
such as the spatial and rotational invariance of both the Hamiltonian and the vacuum. \\
\\
The architecture of the article is as follows:\\
\\
  
In section \ref{Recap} we briefly review elements of the formalism developed in 
 \cite{LLT1,LLT2, LLT3}.
 
In section \ref{Hamiltonian renormalisation} we perform the direct Hamiltonian 
renormalisation for a massive free scalar field in $D+1=3$ dimensions, by breaking it down into several independent renormalisation steps for each direction. It will transpire that this splitting is doable also for more dimensions and that each dimension can be treated 
independently due to a certain factorisation property. We can explicitly compute
the family of fixed point covariances of the Hilbert space measure family  
and find that it perfectly matches, as in the 1+1 case, the perfect Hilbert measure family 
that one obtains by the cylindrical projections of the known continuum Hilbert space measure.
The Hilbert space measure contains the same information as the Hilbert space together with 
its vacuum vector.
By the argument already exploited in the 1+1 dimensional case, the agreement of the fixed point
continuum Hamiltonian with the known continuum Hamiltonian then immediately follows.
 
 In section \ref{Renormalisation with changed RG-map} we investigate the consequences of 
 modifying the {\it coarse graining map} $I_{M \rightarrow2M}$ that we used 
 so far in this series of papers and which relates lattices with $M^D$ vertices to those 
 with $(2M)^D$ many. It was already shown in \cite{LLT2} for the 1+1 dimensional case 
 that not every choice of coarse graining map employed in the literature passes our criterion of being an 
 allowed cylindrically consistent coarse graining map. We find that at least when 
 modifying it to $I_{M\rightarrow M'}$ with $M'=2M,3M,5M,...$ we still
 obtain cylindrically consistent coarse graining
 maps and, moreover, they all lead to the same fixed point Hilbert space measure and Hamiltonian.
 Moreover, it is possible to mix and concatenate different blocking kernels for different
 directions leading to more general coarse graining maps such as those on hypercuboids rather
 than hypercubes and beyond. The fixed point structure is robust under such modifications, thus adding to the degree of universality of the model.  
  
In section \ref{Rotational Invariance} we present how the continuum concept of rotational invariance can be expressed as a condition on the finite resolution Hilbert space measures
$\nu^*_{\epsilon}$ of the fixed point family and and we will numerically demonstrate that in the case of the massive free scalar field the perfect lattice action seems to satisfy this condition.

In section \ref{Conclusion} we summarize and give an outlook to further research directions.

\section{Review of direct Hamiltonian Renormalisation}
\label{Recap}
This section serves to recall the notation and elements of Hamiltonian 
renormalisation from \cite{LLT1, LLT2} to which the reader is referred to for all the details.\\
\\
We consider infinite dimensional conservative Hamiltonian systems on globally hyperbolic 
spacetimes of the form $\mathbb{R}\times \sigma$.  If  $\sigma$ is not compact
 one introduces an infrared (IR) cut-off $R$ for the spatial manifold $\sigma$ by restricting to test-functions which are defined on a compact submanifold, e.g. a torus $\sigma_R:=[0,R]^D$ if $\sigma=\mathbb{R}^D$. We will assume this cut-off $R$ to be implicit in all formulae 
 below but do not display it to keep them simple, see \cite{LLT1,LLT2} for the explicit
 appearance of $R$.

Moreover, an ultraviolet (UV) cut-off $\epsilon_M:=R/M$ is introduced by restricting the smearing functions $f_M$ to a finite spatial resolution. In other words, $f_M\in L_{M}$ is defined by its value on the vertices of a graph, which we choose here to be a cubic lattice, i.e. there are $M^D$ many vertices, labelled by $m\in\mathbb{Z}^D_M$, $\mathbb{Z}_M=\{0,1,...,M-1\}$. In this 
paper we consider a background dependent (scalar) QFT and thus we have access to a natural 
inner product defined by it. For more general theories this structure is not available
but the formalism does not rely on it. The scalar products for $f_M,g_M\in L_{M}$ and respectively for $f,g\in L=C^{\infty}([0,R])$ are defined by
\begin{align}
\langle f_M, g_M \rangle_M = \epsilon^D_M\sum_{m\in\mathbb{Z}^D_M} \bar{f}_M(m) g_M(m), \hspace{20pt}
\langle f,g \rangle =\int_{\sigma_R} d^Dx \bar{f}(x)g(x)
\end{align} 
where $\bar{f}$ denotes the complex conjugate of $f$. In other words, $L$ is the space of all compactly supported smooth functions such that for $f,g\in L$ the inner product $\langle f,g\rangle$ stays finite. The similar statement for finite sequences and the inner product $\langle .,.\rangle_M$ defines the Hilbert space $L_M$.\\

Given an $f_M:\mathbb{Z}^D_M\rightarrow \mathbb{R}$ we can embed it into the continuum by  an {\it injection map} $I_M$
\begin{equation} \label{injection map}
\begin{split}
I_M : \hspace{10pt} L_{M} &\rightarrow \hspace{10pt} L \\
  f_M &\mapsto (I_M f_M ) (x):=\sum_{m\in \mathbb{Z}^D_M}f_M(m)\chi_{m\epsilon_M}(x)
\end{split}
\end{equation}
with $\chi_{m\epsilon_M}(x):=\prod_{a=1}^D\chi_{[m^a\epsilon_M,(m^a+1)\epsilon_M)}(x)$ being the characteristic function over the displayed intervals. Note that indeed the coefficient $f_M(m)$ is the value of $I_M f_M$ at $x=m\epsilon_M$.

$L$ is much larger than the range of $I_M$. This allows us to define its corresponding left inverse: {\it the evaluation map} $E_M$ is found to be
\begin{equation} \label{evaluation map}
\begin{split}
E_M : \hspace{10pt} L &\rightarrow \hspace{10pt} L_{M} \\
  f &\mapsto (E_M f ) (m):=f(m\epsilon_M)
\end{split}
\end{equation}
and by definition obeys
\begin{align}\label{left-inverse}
E_M\circ I_M= \text{id}_{L_M}
\end{align}
where $\text{id}_{L_M}$ denotes the identity map on the space $L_M$. 
Given those maps we are now able to relate test functions and thus also observables from the continuum with their discrete counterpart, e.g. for a smeared scalar field one defines: 
\begin{align}
\phi_M[f_M]:=\langle f_M, \phi_M\rangle_M,\hspace{20pt}\phi_M(m):=(I^\dagger_M\phi)(m)=\int_{[0,R)^D}d^Dx\; \chi_{m\epsilon_M}(x)\phi(x)
\end{align}
Indeed, since the kernel of any map $C: L\rightarrow L$ in the continuum is given as
\begin{align}
\langle f, C g\rangle =: \int_{[0,R]^D}d^Dx\int_{[0,R]^D}d^Dy \;C(x,y)\bar{f}(x)g(y)
\end{align}
it follows
\begin{align}
\langle I_M f_M, C I_M g_M\rangle = \langle f_M, [I^{\dagger}_M C I_M]g_M\rangle_M=: \langle f_M, C_M g_M\rangle_M
\end{align}
which shows that
\begin{align}
C_M(m,m')=\epsilon^{-2D}_M\langle \chi_{m\epsilon_M},C \chi_{m'\epsilon_M}\rangle
\end{align}\\

The concatenation of evaluation and injection for different discretisations shall be called {\it coarse graining map} $I_{M\rightarrow M'}$ if $M<M'$:
\begin{align}
I_{M\rightarrow M'} = E_{M'} \circ I_M : L_{R,M} \rightarrow L_{R,M'}
\end{align}
as they correspond to viewing a function defined on the coarse lattice as a function on a finer lattice of which the former is not necessarily a sublattice although. In practice we will 
choose the set of $M$ such that it defines a partially ordered and directed index set.
The coarse graining map is a free choice of the renormalisation process whose flow it drives, and its viability can be tested only a posteriori if we found a fixed pointed theory which agrees with the measurements of the continuum. Hence proposals for such a map should be checked at least in simple toy-models before one can trust their predictions.

The coarse graining maps are now used to call a family of 
of measures $M\mapsto \nu_M$ on a suitable space of fields $\phi_M$ {\it cylindrically 
consistent} iff 
\begin{align}\label{cylindricalconsistency}
\nu_{M}(w_M[f_M])=\nu_{M'}(w_{M'}[I_{M\rightarrow M'}\circ f_M])
\end{align}
where $w_M$ is a Weyl element restricted to the configuration degrees of freedom, i.e. for a 
scalar field theory as in the present paper 
$w_M[f_M]=\exp(i\phi_M[f_M])$. The measure $\nu_M$  can be considered as the positive linear 
GNS functional on the Weyl $^\ast-$algebra generated by the $w_M[f_M]$ with GNS data 
$({\cal H}_M,\Omega_M)$, that is
\begin{align}
\nu_M(w_M[f_M])=\langle \Omega_M , w_M[f_M]\Omega_M\rangle_{{\cal H}_M}
\end{align}
In particular, the span of the $w_M[f_M]\Omega_M$ lies dense in $\mathcal{H}_M$ and 
we simplify the notation by refraining from displaying a possible GNS null space. 
Under suitable conditions \cite{Yam75} a cylindrically consistent family has a continuum measure $\nu$ as a projective limit
which is related to its family members $\nu_M$ by 
\begin{align} \label{projectivelimit}
\nu_M(w_M[f_M])=\nu(w[I_M f_M]) 
\end{align}
It is easy to see that (\ref{projectivelimit}) and (\ref{cylindricalconsistency}) are compatible
iff we constrain the maps $I_M, E_M$ by the condition for all $M<M'$
\begin{align}\label{cylconrel}
I_{M'}\circ I_{M\rightarrow M'} = I_M
\end{align}
This constraint which we also call {\it cylindrical consistency} means that injection into the continuum 
can be done independently of the lattice on which we consider the function to 
be defined, which is a physically plausible assumption. 

 In the language of the GNS data  $(\mathcal{H}_M,\Omega_M)$ cylindrical consistency means that 
 the maps 
 \begin{align}
J_{M\rightarrow M'} :\hspace{20pt} \mathcal{H}_M\hspace{20pt} & \rightarrow \hspace{20pt}\mathcal{H}_{M'}\\
w_M[f_M]\Omega_M & \mapsto w_{M'}[I_{M\rightarrow M'} f_M]\Omega_{M'}
\end{align}
define a family of {\it isometric} injections of Hilbert spaces. The continuum GNS data 
are then given by the corresponding inductive limit, i.e. the embedding of Hilbert spaces 
defined densely by $J_M w_M[f_M]\Omega_M=w[I_M f_M] \Omega$ which is isometric. 
Note that $J_{M'} J_{M\rightarrow M'}=J_M$. 
The GNS data are completed by a family of positive 
self-adjoint Hamiltonians $M\mapsto H_M$ defined 
densely on the $w_M[f_M]\Omega_M$ to the OS data $({\cal H}_M, \; \Omega_M,\; H_M)$.
We define a family of such Hamiltonians to be cylindrically consistent provided that
\begin{align}
J_{M\rightarrow M'}^\dagger H_{M'} J_{M\rightarrow M'}=H_M
\end{align}
It is important to note that this does {\it not} define an inductive system of operators which 
would be too strong to ask and thus 
does not grant the existence of a continuum Hamiltonian. However, it grants the existence 
of a continuum positive quadratic form densely defined by 
\begin{align}
J_M^\dagger H J_M =H_M
\end{align}
If the quadratic form can be shown to be closable, one can extend it to a positive self-adjoint 
operator.    

In practice, one starts with an {\it initial} family of OS data $({\cal H}^{(0)}_M,\;
\Omega^{(0)}_M,\; H^{(0)}_M)$ usually obtained by some {\it naive} discretisation 
of the classical Hamiltonian system and its corresponding quantisation. The corresponding
GNS data will generically not define a cylindrically consistent family of measures, i.e. 
the maps $J_{M\rightarrow M'}$ will fail to be isometric. Likewise, the family 
of Hamiltonians will generically fail to be cylindrically consistent. Hamiltonian 
renormalisation now consists in defining a sequence of an {\it improved} OS data 
family $n\mapsto 
 ({\cal H}^{(n)}_M,\; \Omega^{(n)}_M,\;H^{(n)}_M)$ defined inductively by 
 \begin{align} \label{improving}
J^{(n)}_{M\rightarrow M'} w_M[f_M] \Omega^{(n+1)}_M:=w_{M'}[I_{M\rightarrow M'} f_M]
\Omega^{(n)}_{M'},\;\;H^{(n+1)}_M:=J^{(n)}_{M\rightarrow M'} H^{(n)}_{M'} J^{(n)}_{M\rightarrow M'}
\end{align}
Note that $H^{(n)}_M \Omega^{(n)}_M=0$ for all $M,n$. 
If the corresponding flow (sequence) has a fixed point family $({\cal H}^\ast_M,\; \Omega^\ast_M,\;
H^\ast_M)$ then its internal cylindrical consistency is restored by construction.

\section{Hamiltonian renormalisation of the massive free quantum scalar field}\label{Hamiltonian renormalisation}

In \cite{LLT2} the Hamiltonian renormalisation prescription introduced above has been tested with a model whose continuum theory is known and hence presents a way to validate the method: the massive free quantum scalar field described by the action
\begin{align}
S:=\frac{1}{2\kappa}\int_{\mathbb{R}^{D+1}} dt d^Dx [\frac{1}{c}\dot{\phi}^2-c\phi\omega\phi]
\end{align}
with ($n=1,2,..$)
\begin{align}\label{contcovariance}
\omega^2=\omega^2(p,\Delta)=\frac{1}{p^{2(n-1)}}(-\Delta+p^2)^n
\end{align}
where $p=\frac{mc}{\hbar}$ is the inverse Compton length. Following \cite{LLT2} we will study here the Poincare invariant case with $n=1$, other models with $n\not=1$ can be studied 
with the methods developed for $n=1$ by contour integral techniques.

After performing the Legendre transform, introducing the IR cut-off and discretising the theory for various scales $M$ one considers the lattice Hamiltonian family ($\hbar=1$)
\begin{equation}\label{DiscretizedHamiltonian}
H_M:=\frac{c}{2}\sum_{m\in \mathbb{Z}^D_M} \left(\kappa\epsilon^D_M  \pi^2_M(m) +\frac{1}{\epsilon^D_M\kappa}\phi_M(m)(\omega^2_M\cdot\phi_M)(m)\right)
\end{equation}
with ($\pi:=\dot{\phi}/\kappa$)
\begin{align}
\phi_M (m):=\int_{[0,1]^D} d^Dx\; \chi_{m\epsilon_M}(x)\phi(x),\hspace{20pt}\pi_M(m):=(E_M \pi)(m)=\pi(m\epsilon_M)
\end{align}
and $\omega^2_M=\omega^2(p,\Delta_M)$ is to be understood in terms of $\Delta_M$ the naively discretized Laplacian, which reads e.g. in two dimensions:
\begin{equation}
(\Delta^{(0)}_M f_M)(m):=\frac{1}{\epsilon_M^2}\left(f_M(m+e_1)+f_M(m+e_2)+f_M(m-e_1)+f_M(m-e_2)-4f_M(m)\right)
\end{equation}
with $e_i$ being the unit vector in direction $i$. One can write down the explicit action of the coarse graining map for projecting a lattice on a finer version with twice as many lattice points:
\begin{equation}
(I_{M\rightarrow 2M} f_M)(m) =
\sum_{m'\in \mathbb{Z}^D_M}\chi_{m'\epsilon_{2M}}(m\epsilon)f_M(m')
= f_M(\lfloor\frac{m}{2}\rfloor)
\end{equation}
where $\lfloor x \rfloor$ denotes the component wise Gauss bracket. The cylindrical consistency condition (\ref{cylindricalconsistency}) demanded that the measures on both discretisation, $M$ and $2M$, agree. Being a free field theory, one can show that the measure can be written as a Gaussian measure described at the fixed point by a covariance $c^\ast_M$, 
thus (\ref{cylindricalconsistency}) reads explicitly
\begin{align}
e^{-\frac{1}{2}\langle I_{M\rightarrow 2M}f_M,c^*_{2M}I_{M\rightarrow2M}f_M\rangle_{2M}}= e^{-\frac{1}{2}\langle f_M,c^*_M f_M\rangle_M}
\end{align}
Thus by studying the flow defined by
\begin{align}\label{Covarianceflow}
c^{(n+1)}_M := I^{\dagger}_{M\rightarrow2M}c^{(n)}_{2M}I_{M\rightarrow2M}
\end{align}
we know that the existence of a fixed point $c^*_M$ describes a Gaussian measure family, which is equivalent to corresponding Hilbert spaces $\mathcal{H}_M^*$ with vacua $\Omega^*_M$ which are all annihilated by the correspondingly defined Hamiltonians $H^\ast_M$.

\subsection{Determination of the fixed point covariance}

The flow defined by (\ref{Covarianceflow}) may lead to various fixed points (or none at all) depending on the initial family $c^{(0)}_M$. Thus, the naive discretisation should be of such a form that it captures important features of the continuum theory. For example, we will demand the covariance to be translation invariant, which is a property of the discretised Laplacian and will remain true under each renormalisation step.\\
We begin by rewriting (\ref{DiscretizedHamiltonian}) in terms of discrete annihilation and creation operators
\begin{align}
a^{(0)}_M(m):=\frac{1}{\sqrt{2\hbar\kappa}}\left[\sqrt{{\omega^{(0)}_M}/{\epsilon^D_M}}\phi_M -i\kappa\sqrt{{\epsilon^D_M}/{\omega^{(0)}_M}}\pi_M(m)\right]
\end{align}
where 
\begin{align}
[\omega^{(0)}_M]^2:=p^2-\Delta^{(0)}_M
\end{align}
which after some standard algebra displays the Hilbert space measure as:
\begin{align}
\nu^{(0)}_{M}(w_M[f_M])=\nu_{M}\left(e^{i\langle f_M, \phi_{M}\rangle_{M}}\right)=\exp\left(-\frac{1}{4}\langle f_M,\frac{\hbar\kappa}{2}\omega_{M}^{-1}f_M\rangle_M\right)
\end{align}
Hence our starting covariance is given as:
\begin{align}
c^{(0)}_M=\frac{\hbar\kappa}{2}[\omega^{(0)}_M]^{-1}
\end{align}
Using the discrete Fourier transform ($k_M=\frac{2\pi}{M}$)
\begin{align}
f_M(m)=\sum_{l\in\mathbb{Z}^D_M}\hat{f}_M(l)e^{ik_Ml\cdot m},\hspace{20pt}\hat{f}_M(l):=M^{-D}\sum_{m\in\mathbb{Z}^D_M}f_M(m)e^{-ik_Mm\cdot l}
\end{align}
we diagonalise the discretised Laplacian appearing in $\omega^{(0)}_M$. Thus, the initial covariance family becomes in $D=2$ (dropping the factor
 $\frac{2}{\hbar\kappa}$ in what follows)
 \begin{align} 
\hat{c}^{(0)}_M(l)&=\frac{1}{\sqrt{-\frac{1}{\epsilon^2_M}(2\cos(k_Ml_1)+2\cos(k_Ml_2)-4)+p^2}}=\nonumber\\
&=\int_{\mathbb{R}}\; \frac{dk_0}{2\pi} 
\frac{\epsilon_M^2}{[k_0^2+p^2]\epsilon_M^2+(4 -2\cos(k_Ml_1)-2\cos(k_Ml_2))}\label{integrand}
\end{align}
with $l\in\mathbb{Z}^2_M$ and we used the residue theorem.
We rewrite the integrand of (\ref{integrand}) as ($t_i=k_Ml_i, q^2:=(k_0^2+p^2)\epsilon_M^2$)
\begin{equation} \label{startingpoint2D}
\hat{c}^{(0)}_M(k_0,l)=\frac{1}{2}\; \frac{\epsilon_M^2}{[q^2/4+(1-\cos(t_1))]-[-q^2/4-(1-\cos(t_2))]}
\end{equation}
Since $1+q^2/4> \cos(t), \forall p>0,t\in \mathbb{R}$ one deduces that the first of the 
square brackets in (\ref{startingpoint2D}) is always positive, the other one always negative. Consequently, they lie in different halfplanes of $\mathbb{C}$. This can be used to artificially write this as an integral in the complex plane, by inverting the residue theorem: Given $z_1,z_2\in\mathbb{C}2$ with $Re(z_1)>0, Re(z_2)<0$ and a curve $\gamma$ going along $i\mathbb{R}$ from $+i\infty$ to $-i\infty$ and closing in the right plane on a half circle with radius $R\rightarrow \infty$, we can write:
\begin{equation}\label{InvertResThm}
\oint_{\gamma} dz \frac{1}{(z-z_1)(z-z_2)}=2\pi i \frac{1}{z_1-z_2}
\end{equation}
since the integrand decays as $z^{-2}$ on the infinite half circle.
We have chosen the orientation of $\gamma$ counter clock wise. Note that this seemingly breaks the symmetry between $t_1$ and $t_2$. However, this is only an intermediate artefact of the free choice of $\gamma$ which will disappear at the end of the computation.
 
Substituting $z\rightarrow z/2$ the initial covariance can thus be written  
\begin{equation} \label{startingdecoupledcovariance}
\hat{c}^{(0)}_M(l)=-\oint_{\gamma} dz \frac{1}{8\pi i}\; 
\frac{\epsilon^2_M}{\epsilon^2_M(\frac{p^2+k_0^2}{2}-z)/2+1-\cos(t_1)}\;\;\frac{\epsilon^2_M}{\epsilon^2_M(\frac{p^2+k_0^2}{2}+z)/2+1-\cos(t_2)}
\end{equation}
In order to shorten our notation, we will introduce: $q^2_{1,2}(z):=
\epsilon^2_M([k_0^2+p^2]/2\mp z)$.
The starting point of our RG flow is now factorised into two factors which very closely resemble the 1+1 dimensional case. This is the promised factorising property.\\

 On the other hand, one can also show by explicit calculation (see appendix \ref{sa}) that both directions decouple in the renormalisation transformation (\ref{Covarianceflow}). Since the initial covariance factorises  under the contour integral over $\gamma$
this factorisation is preserved under the flow and implies that the flow of the covariance in each direction can be performed separately. 
At the end we then compute the resulting integral over $z$ along $\gamma$. In addition, the decoupling of the flow (\ref{Covarianceflow}) and the factorisation of 
the initial family of covariances (\ref{startingdecoupledcovariance}) for the naive discretisation
of the Laplacian are features that occur independently
of the dimension $D$. For the decoupling this follows immediately from the corresponding
generalisation of (\ref{flowdefinition}) as the sum over $\delta',\delta^{\prime\prime}$ is carried
out on the exponential function which contains both linearly in the exponent. For the factorisation
we note the following iterated integral identity for complex numbers $k_j,\; j=1,..,D$ with strictly 
positive real part  
\begin{align}
\frac{1}{k_1+..+k_D}=(2\pi i)^{D-1}
\oint_\gamma\; \frac{dz_1}{z_1-k_1}\;
\oint_\gamma\; \frac{dz_2}{z_2-k_2}\;..
\oint_\gamma\; \frac{dz_{D-1}}{z_{D-1}-k_{D-1}}\;\frac{1}{z_1+..+z_{D-1}+k_D}
\end{align}
in which $\gamma$ is always the same closed contour with counter clock orientation over the
half circle in the positive half plane followed by the integral over the imaginary axis. Because of that
the real part of each of the integration variables $z_j$ is non negative so that the last fraction has 
a denominator with strictly positive real part. Accordingly, the only pole of the integrand for the $z_j$ integral in the domain bounded by $\gamma$ is $k_j$ and the claim follows from the residue theorem. It transpires that the strategy illustrated for the case $D=2$ also solves the case of general $D$ and it therefore suffices to carry out the details for $D=2$.

The flow now acts on the integrand of the contour integral and we can do it for each $z$ 
separately.
The flow in each direction is thus described by exactly the same map as in the one dimensional case in \cite{LLT2}. We can therefore immediately copy the fixed point covariance from there.
We just have to keep track of the $z$ dependence. 
In direction $i=1,2$ the covariance can be parametrised by 
three functions of $q_i(z)$  ($t_i=k_Ml_i$, $l_i\in\mathbb{Z}_M$)
\begin{align}
\hat{c}_{M}^{(n)}(k_0,l_i,z)=\frac{\epsilon_M^2}{q_{i}^3(z)}\frac{b_n(q_{i}(z))+c_n(q_{i}(z))\cos(t_i)}{a_n(q_{i}(z)-\cos(t_i)}
\end{align}
The initial functions are 
\[a_0(q_{1,2}(z))=1+\frac{q^2_{1,2}}{2},\hspace{15pt} b_0(q_{1,2}(z))=\frac{q^3_{1,2}}{2},\hspace{15pt} c_0(q_{1,2}(z))=0\]
Before plugging in the fixed points, however, one has to check whether the flow will drive the starting values into a finite fixed point, i.e. all the numerical prefactors that are picked up in front of the covariance should cancel each other. Indeed, one RG steps corresponds to
\begin{align}
(2\pi i) \hat{c}^{(n+1)}_M(k_0,l)&=-\frac{1}{4}\oint_{\gamma} dz \left(\sum_{\delta_1=0,1}(1+\cos(k_{2M}(l_1+\delta_1M)))\hat{c}^{(n)}_{2M}(k_0,l_1,z)\right) \times\\
&\hspace{120pt}
\times
\left(\sum_{\delta_2=0,1}(1+\cos(k_{2M}(l_2+\delta_2M)))\hat{c}^{(n)}_{2M}(k_0,l_2,-z)\right)\nonumber
\end{align}
Note that $\epsilon_M\rightarrow \epsilon_{2M}=\epsilon_M/2$ whence 
\begin{equation} \label{transformqpm}
q^2_{1,2}:=\epsilon^2_M(\frac{p^2+k_0^2}{2}\mp z) \rightarrow 
\frac{1}{4}\epsilon^2_M(\frac{p^2+k_0^2}{2}\mp z) = q^2_{1,2}/4
\end{equation}
Collecting all powers of $2$,  we get 1. minus two from the $\epsilon^2_M$ in the numerator 
of the factor for both directions, that is altogether minus four; 2. the RG map gives an additional 
minus two because of the $1/4$ prefactor; and 3. due to (\ref{transformqpm})  
the factor $q_1^{-3}q_2^{-3}$ gives a power of plus six. Hence the overall power of two is zero.

Accordingly, we find  the same fixed points as in \cite{LLT2}:
\begin{align}
a^*(q_{1,2}) &=\text{ch}(q_{1,2})\\
b^*(q_{1,2}) &=q_{1,2} \text{ch}(q_{1,2})-\text{sh}(q_{1,2})\\
c^*(q_{1,2}) &=\text{sh}(q_{1,2})-q_{1,2}
\end{align}
where we write shorthand for the hyperbolic functions: $\text{ch}(q):=\cosh(q)$ and $\text{sh}(q):=\sinh(q)$. Thus we find with $t_j=k_M l_j$ 
\begin{align}\label{finalIntegralof2DRen}
\hat{c}^*_M(k_0,l)&=-\left(\frac{\epsilon_M^4}{2\pi i}\right)\; \oint_\gamma\; dz\;
\prod_{j=1,2} \frac{1}{q_j^3}
\frac{q_j \text{ch}(q_j) - \text{sh}(q_j) +(\text{sh}(q_j)-q_j)\cos(t_j)}{\text{ch}(q_j)-\cos(t_j)}
\end{align}
Note that it is not necessary to pick a square root of the complex parameter 
$q^2_{1,2}(z)=\epsilon_M^2(\frac{k_0^2+p^2}{2}\mp z)$ since the integrand only depends 
on the square, despite its appearance (in other words, one may pick the branch arbitrarily,
the integrand does not depend on it). It follows that the integrand is a single valued 
function of $z$
which is holomorphic everywhere except for simple poles which we now determine, and which 
allow to compute the contour integral over $\gamma$ using the residue theorem (see appendix \ref{sa} for further details) and end up with 
\begin{align}\label{CovarianceResult2D}
\hat{c}_M^*(k_0,l)=& \epsilon^2_M\frac{[q_N {\rm ch}(q_N)-{\rm sh}(q_N)]+[{\rm sh}(q_N)-q_N]\cos(t_2)}{q_N^3[{\rm ch}(q_N)-\cos(t_2)]}+\\
&\hspace{50pt}-2\epsilon^2_M\underset{N\in\mathbb{Z}}{\sum}\frac{\cos(t_1)-1}{(2\pi N +t_1)^2}
\frac{1}{q_N^3} \;\frac{q_N \text{ch}(q_N)-\text{sh}(q_N)+(\text{sh}(q_N)-q_N)\cos(t_2)}
{\text{ch}(q_N)-\cos(t_2)}\nonumber
\end{align}
The result has no manifest symmetry in $t_1\leftrightarrow t_2$ but from the derivation 
it is clear that it must be. Note that each term in the sum remains finite for $\epsilon\rightarrow 0$ as the individual parts contribute inverse powers of $\epsilon_M$:  $(\cos(t)-1)$ scales as $\mathcal{O}(\epsilon^2_M)$, since 
$t=k_R  \epsilon_M l$ depends linearly on $\epsilon_M$ as well as $q=\epsilon^2_M (p^2+k_0)^2$. Thus $(q^2+(t+2\pi N)^2)=\mathcal{O}(\epsilon^2_M)$ if $N=0$ or approaches a constant else.

\subsection{Consistency check with the continuum covariance}

The mere existence of a fixed point measure family described by the covariance (\ref{CovarianceResult2D}) of the flow induced by (\ref{Covarianceflow}) does not necessarily mean that it has any relation with the known continuum theory. We will thus invoke the consistency check also presented in \cite{LLT2}, which consists of looking at the cylindrical projection at resolution $M$ of the continuum covariance $c:=\frac{1}{2}\omega^{-1}$ in $D=2$. Using that the latter is given by (\ref{contcovariance}) we find its projection to be
\begin{align}
c_M(m,m')&=\epsilon_M^{-4} (I^\dagger_M\; c\; I_M)(m,m') =\\
&
=\epsilon^{-4}_M \int^{(m_1+1)\epsilon_M}_{m_1\epsilon_M}dx_1
\int^{(m_2+1)\epsilon_M}_{m_2\epsilon_M}dx_2\int^{(m'_1+1)\epsilon_M}_{m'_1\epsilon_M}dy_1\int^{(m'_2+1)\epsilon_M}_{m'_2\epsilon_M}dy_2 \hspace{5pt}\;c(x,y)\nonumber
\end{align}
see \cite{LLT2} for more details. Using that the $e_{n}(x):=\frac{1}{R} e^{ik_R n\cdot x},
\; k_R=2\pi/R$ 
define an orthonormal basis of $L_R=L_2([0,R]^2,d^2x)$ one finds the resolution of identity
\begin{equation}
\frac{1}{R^2}\sum_{n\in \mathbb{Z}^2}e^{ik_R (x-y)\cdot n}=\delta_R(x,y):=\delta_R(x_1,y_1)\delta_R(x_2,y_2)
\end{equation}
We use this to write the covariance as
\begin{align}
c(x,y)&= \frac{1}{2}\left(-\Delta_{Rx}+p^2\right)^{-1/2} \delta_R(x,y)=
\int \; \frac{dk_0}{2\pi} \left(-\Delta_{Rx}+k_0^2+p^2\right)^{-1} \delta_R(x,y)\\
 &=\int \frac{dk_0}{2\pi}
 \sum_{n\in\mathbb{Z}^2} e_{nR}(y)\left(-\Delta_{Rx}+p^2+k_0^2\right)^{-1}e_{nR}(x)=
 \frac{1}{R^2}\sum_{n\in\mathbb{Z}^2}e^{ik_R n\cdot (x-y)}\frac{1}{n^2 k_R^2+k_0^2+p^2}\nonumber
\end{align}
Now we can perform the integrals, e.g.
\begin{equation}
\int^{(m_1+1)\epsilon_M}_{m_1\epsilon_M}dx_1 e^{i(2\pi) n_1 x_1}=\frac{1}{ik_R n_1}\left(e^{ik_R n_1(m_1+1)\epsilon_M}-e^{i k_R n_1m_1\epsilon_M}\right)
\end{equation}
where the case $n_1=0$ is obtained using de l'Hospital. We find with $k_M=2\pi/M$
\begin{align}
c_M(m&,m')=\epsilon^{-4}_M \frac{1}{R^2}\int \frac{dk_0}{2\pi}\; 
\sum_{n\in\mathbb{Z}^2}\frac{1}{n^2 k_R^2+p^2+k_0^2}\left(\int^{(m_1+1)\epsilon_M}_{m_1\epsilon_M}dx_1 e^{i(2\pi)n_1x_1}\right)\times\\
&\hspace{20pt}\times\left(\int^{(m_2+1)\epsilon_M}_{m_2\epsilon_M}dx_2 e^{i(2\pi)n_2x_2}\right)\left(\int^{(m_1'+1)\epsilon_M}_{m_1'\epsilon_M}dy_1 e^{i(2\pi)n_1y_1}\right)\left(\int^{(m_2'+1)\epsilon_M}_{m_2'\epsilon_M}dy_2 e^{i(2\pi)n_2y_2}\right)\nonumber\\
&=R^{-2} \int\frac{dk_0}{2\pi} \sum_{n\in\mathbb{Z}^2}\frac{1}{n^2 k_R^2+p^2+k_0^2}\;e^{ik_Mn\cdot (m-m')}\;
\frac{4}{k_M^4 n_1^2 n_2^2}[1-\cos(k_M n_1)]\;[1-\cos(k_M n_2)]\nonumber
\end{align} 
We now wish to proceed exactly as in \cite{LLT2} and thus split the sum over $n_j=l_j+M N_j$ with 
$l_j\in \mathbb{Z}_M$ and $N\in \mathbb{Z}^2$
\begin{align}
c_M(m,m')&=R^{-2} \epsilon_M^4 \int\frac{dk_0}{2\pi} \sum_{m\in\mathbb{Z}_M^2}
e^{ik_Ml \cdot (m-m')}\sum_{N\in \mathbb{Z}^2}\times\;\\
& \;\;\;\;\times\frac{[1-\cos(k_M (l_1+M N_1)]
\;[1-\cos(k_M (l_2+M N_2)]}{(l+M N)^2 k_M^2+\epsilon_M^2(p^2+k_0^2)}\;
\;
\frac{4}{k_M^4 (l_1+M N_1)^2  (l_2+M N_2)^2}\nonumber
\end{align}
from which we read off the Fourier transform of $c_M(m)=R^{-2}\sum_{l\in \mathbb{Z}_M^2}
e^{k_M l\cdot m} \hat{c}_M(l)$
\begin{align}
\hat{c}_M(k_0,l)&=\epsilon_M^4
\sum_{N\in \mathbb{Z}^2}\;\times \\
&\;\;\;\;\times\frac{[1-\cos(k_M (l_1+M N_1)]
\;[1-\cos(k_M (l_2+M N_2)]}{(l+M N)^2 k_M^2+q^2}\;
\;
\frac{4}{k_M^4 (l_1+M N_1)^2  (l_2+M N_2)^2}\nonumber
\end{align}
Using the contour integral idea as in the previous section we obtain
\begin{align}
\hat{c}_M(k_0,l)=&-\frac{1}{2\pi i} \oint_\gamma\; dz\;
\prod_{j=1,2}\; 
\\
&\hspace{30pt}\big[\sum_{N_j \in \mathbb{Z}} \frac{\epsilon_M^2}{(l_j+M N_j)^2 k_M^2+q_j(z)^2}\;
\;
\frac{2}{k_M^2 (l_j+M N_j)^2}[1-\cos(k_M (l_j+M N_j))]\big]\nonumber
\end{align}
where $q_j(z)$ is the same as in the previous subsection. Now the two sums of the formula 
above are exactly the same that occurred in \cite{LLT2} with $q^2$ replaced by $q_j(z)^2$ and 
$t$ replaced by $t_j =l_j k_M$. Thus, we can copy the result from there and find 
\begin{align} \label{continuumcovariance}
\hat{c}_M(k_0,l)=- \frac{\epsilon_M^4}{2\pi i} \oint_\gamma\; dz\;
\prod_{j=1,2}\; \frac{1}{q_j^3}\frac{q_j(z) \text{ch}(q_j)-\text{sh}(q_j)
+[\text{sh}(q_j)-q_j]\cos(t_j)}{\text{ch}(q_j)-\cos(t_j)}
\end{align}
with $q_j\equiv q_j(z)$.
Comparing (\ref{continuumcovariance}) and (\ref{finalIntegralof2DRen}) we see that both agree,
thus the fixed point covariance family indeed coincides with the continuum covariance family.

\section{Fixed points of the free scalar field for changed RG-flows}\label{Renormalisation with changed RG-map}

The aim of this section is to change the block-spin-transformation we have used so far
and to check whether the fixed point is changed as well.
As has already been discussed in \cite{LLT3} not every coarse graining map fulfils the cylindrical consistency relation which induces a corresponding relation on the family of coarse grained 
measures. Note that coincidence of continuum measures with their cylindrical 
(finite resolution) projections can only be deduced if one uses the same blocking kernel 
(which defines those projections).
Thus, it a natural question to ask whether other maps of the kind $I_{M\rightarrow M'}$ apart from $M'=2M$ will also lead to physically relevant theories. Due to the cylindrical consistency property of $I_{M\rightarrow M'}$ it is apparent that this is the case for all $M'=2^nM$ for $n\in\mathbb{N}$. 
A natural extension would be to consider powers of any prime number. 
In this section we present how at least for the choice for $M'=3M$ and $M'=5M$ this indeed gives the same fixed point covariance and argue  that it should hold true for every choice of 
prime number. This would be useful because the set $\mathbb{N}$ is partially ordered and 
directed by $<$ but given $m_1,m_2\in \mathbb{N}$ we do not always find $m_3>m_1, m_2$ 
with $m_3=m_1 2^{n_1}=m_2 2^{n_2}$.

If one considers $I_{M\rightarrow u M}$ with $u\in\mathbb{P}$ a prime number
then the coarse graining map is given by
\begin{equation}
[I_{M\rightarrow uM}f_M](m)=f_M(\lfloor \frac{m}{u}\rfloor)
\end{equation}
where $\lfloor .\rfloor$ is the component wise Gauss bracket. This map is easily checked to 
be cylindrically consistent $I_{u^k M\to u^{k+l} M}\circ I_{M\to u^k M}=I_{M\to u^{k+l} M}$.
To see this, we note that $\lfloor m/u^k\rfloor =m'$ if $m=m' u^k+r,\; r=0, ... u^k-1$ so that
$\lfloor\lfloor m/u^l\rfloor /u^k\rfloor =m'$ for $\lfloor m/u^l\rfloor =m' u^k+r,\; k=0,..,u^k-1$ that is 
for $m=(m' u^k+r) u^l+s,\; s=0,..u^l-1$ i.e. $m= m' u^{k+l}+t,\; t=0,..,u^{k+l}-1$ i.e. 
$m'=\lfloor m/u^{k+l}\rfloor$.
 
We now use these maps on our Gaussian example. For their covariances 
this implies 
\begin{align}
&\langle f_M, C^{(n+1)}_M f_M\rangle = \epsilon^{2D}_M\sum_{m'_1,m'_2\in\mathbb{Z}^D_M} f_M(m'_1)f_M(m'_2)C^{(n+1)}_M(m'_1,m'_2)\\
&=\langle I_{M\rightarrow uM} f_M, C^{(n)}_{uM} I_{M\rightarrow uM}f_M\rangle 
=\frac{\epsilon^{2D}_{uM}}{u^{2D}}\sum_{m'_1,m'_2\in\mathbb{Z}^D_M}f_M(m'_1)f_M(m'_2)\sum_{\scriptsize \begin{array}{c}
\lfloor m_1/u\rfloor =m'_1,\\
\lfloor m_2/u\rfloor =m'_2
\end{array}
}C^{(n)}_{uM}(m_1,m_2)\nonumber
\end{align}
This allows to deduce by direct comparison:
\begin{align}
C^{(n+1)}_M(m'_1,m'_2)=u^{-2D}\sum_{\delta',\delta''\in\{0,1,..,u-1\}^D} C^{(n)}_{uM}(um'_1+\delta', um'_2+\delta'')
\end{align}
Again we employ translational invariance, i.e. $C^{(n)}_M(m_1,m_2)=C^{(n)}_M(m_1-m_2)$ and find for the Fourier transform: ($k_M=\frac{2\pi}{M}=uk_{uM}$)
\begin{align}
C^{(n+1)}_M(m'_1,m'_2)=
u^{-2D} \sum_{l'\in\mathbb{Z}^D_M}e^{ik_Ml'(m'_1-m'_2)}\sum_{\delta',\delta'',\delta\in\{0,1,...,u-1\}^D}\hat{C}^{(n)}_{3M}(l'+\delta M)e^{ik_{uM}(l'+\delta)\cdot(\delta'-\delta'')}
\end{align}
whence
\begin{equation}\label{generaldecoupling}
\hat{C}^{(n+1)}_M(l')=u^{-2D}\sum_{\delta\in\{0,1,...,u-1\}^D}\hat{C}_{uM}^{(n)}(l'+\delta M)\prod_{i=1}^D\frac{\sin(\frac{u}{2}k_{uM}(l'_i+\delta_i M))^2}{\sin(\frac{1}{2}k_{uM}(l'_i+\delta_i M))^2}
\end{equation}
where we have used that the exponentials decouple, and that the geometric series can be performed explicitly
\begin{align}\label{geometricseries}
\sum_{\delta,\delta'\in\{0,...,u-1\}}e^{ia(\delta-\delta')} = \frac{1-e^{iau}}{1-e^{ia}}\frac{1-e^{-iau}}{1-e^{-ia}}=\frac{\sin(\frac{u}{2}a)^2}{\sin(\frac{1}{2}a)^2}
\end{align}
Since (\ref{generaldecoupling}) states that the flow decouples in general and since 
we can write the initial covariance also in a decoupled form, this allows us to limit our further analysis to the $D=1$ case without loss of generality.\\

In appendix \ref{sb} the determination of the fixed point for the prime $u=3$ will be explicitly performed as it illustrates what needs to be done in the general case. The initial data of the RG-flow is given for $D=1$ with $t=k_Ml, q^2=\epsilon_M^2(k_0^2+p^2)$ by
\begin{equation}\label{3start}
\hat{c}^{(0)}_M(k_0,l)=\frac{\epsilon^2_M}{2(1-\cos(t))+q^2}
\end{equation}
With the help of trigonometric identities, one manages to write the $\hat{c}^{(n)}$ in the form
\begin{equation} \label{3start2}
\hat{c}^{(n)}_M(k_0,l)=\frac{\epsilon^2_M}{q^3}\frac{b_n(q)+c_n(q)\cos(t)}{a_n(q)-\cos(t)}
\end{equation}
with suitably chosen functions $a_n,b_n,c_n$ of $q$ as we already know is true for (\ref{3start}). As it will transpire in appendix \ref{sb}, one finds exactly the same fixed point under the $M\rightarrow3M$ coarse graining map as we found for the $M\to 2M$ 
coarse graining map!

We did the same calculations also for the prime $u=5$ which is considerably more work
but all steps are literally the same and also the fixed point is the same. For reasons of space, we do not display these calculations
here and leave it to the interested reader as an exercise. For the general prime we do 
not have a proof available yet but hope to be able to supply it in a future publication.
However, we do not expect any changes. In any case, for whatever primes the 
fixed point stays the same (it holds at least for $u=2,3,5$) the statement is also true 
for all dimensions due to the factorising property. This factorising property
also makes it possible to study 
in higher dimensions more complicated hypercuboid  like coarse graining block transformations 
rather than hypercube like ones.  In order to illustrate this, we give some details 
for the case $D=2$ dimensions of a rectangle blocking with $u_1=2$ for the first direction 
and  $u_2=3$ for the second. The map is consequently $I_{(M_1,M_2)\rightarrow (2M_1,3M_2)}
=I_{M_1\to 2M_1}\times I_{M_2\to 3 M_2}$.
The naively discretised Laplacian on a lattice with different spacings $\epsilon_{M_1},\epsilon_{M_2}$ is given as (here: $2\epsilon_{M_1}=3\epsilon_{M_2}$)
\begin{align}
&\left(\Delta_M f_M\right) (m):=\\
&=\frac{1}{\epsilon^2_{M_1}}\left(f_M(m+e_1)+f_M(m-e_1)-2f_M(m)
\right)+\frac{1}{\epsilon^2_{M_2}}\left(f_M(m+e_2)+f_M(m-e_2)-2f_M(m)\right)\nonumber
\end{align}
Hence the same strategy from (\ref{InvertResThm}) works again and gives us:
\begin{align}
&\hat{C}^{(0)}_M(k_0,l)=\left(-\frac{1}{\epsilon^2_{M_1}}(2\cos(k_{M_1}l_1)-2)-\frac{1}{\epsilon^2_{M_2}}(2\cos(k_{M_2}l_2)-2)+p^2+k_0^2\right)^{-1}\\
&=-\frac{1}{2^3\pi i}\oint dz \frac{\epsilon^2_{M_1}}{\epsilon^2_{M_1}(z+k_0^2+p^2/2)/2+1-\cos(k_{M_1}l_1)}\;\frac{-\epsilon^2_{M_2}}{\epsilon^2_{M_2}(-z+k_0^2+p^2/2)/2+1-\cos(k_{M_2}l_2)}\nonumber
\end{align}
So both directions decouple and yield, as already shown the same fixed point! It remains to compute the integral which is exactly the same as (\ref{finalIntegralof2DRen}).

A further immediate consequence is that at this fixed point, one could also consider the flow of arbitrary concatenations of different coarse-graining maps, independently for each direction, 
e.g. $.....I_{6M\rightarrow 12M}I_{2M\rightarrow 6M}I_{M\rightarrow2M}$ and we see 
that all of them have the same fixed point. We conclude that the fixed point is quite robust 
under rather drastic changes of the coarse graining map.

\section{Rotational Invariance of the lattice fixed pointed theory}\label{Rotational Invariance}

We now turn our attention towards the much discussed question of {\it rotational invariance} \cite{HN93,LR82,DS12}. By this we mean that most Hamiltonians for continuum theories 
on Minkowski space have $SO(D)$ as a symmetry group besides spatial translation
invariance. On the one hand, a fixed lattice certainly breaks rotational invariance and in the 
case of a hypercubic lattice reduces the invariance to rotations by multiples of $\pm\pi/2$ around the coordinate axes. On the other hand, it is clear that the cylindrical projections of a rotationally invariant measure in the continuum with respect to smearing functions adapted to the 
family of hypercubic lattices in question must carry an imprint of that continuum rotation
invariance. In other words, there must exist a criterion at finite lattice resolution, whether the 
corresponding lattice measure qualifies as the cylindrical projection of a continuum rotationally invariant measure.  

In this section  we identify such a notion of {\it rotational invariance} at finite resolution at least 
for the case of scalar field theories. We will consider the generalisation to other field
contents in future publications. We then successfully test this criterion for the fixed point covariance $\hat{c}^*_M(l)$ from section \ref{Hamiltonian renormalisation}  for the free Klein Gordon field. Due to the factorisation property and due to the possibility of presenting any rotation in terms a composition of rotations about the coordinate axes, we can reduce our attention to two spatial dimensions. This presents an example for how the Hamiltonian renormalisation scheme is able to detect the restoration of continuum properties of the classical theory which upon naive regularisation were lost in the quantisation process.

\subsection{The lattice rotational invariance condition}

Given the IR-restricted compact submanifold of $\sigma$, i.e. the $D$-dimensional torus $\sigma_R$ with periodic boundary conditions and length $R$, one must be precise 
what one means by rotations. Due to the periodicity, the definition of what is understood as a rotation may vary for points which have a distance to the centre of rotation greater then $R/2$. For the convenience of the reader, we will hence present the following description of what will be understood as a rotation in this paper:\\
In order to rotate the system around $x_0\in\sigma_R$ one uses the Euclidian metric on the torus to identify all points as a set $S_r$ which have distance $r>0$ to the central point $x_0$. We then choose $S_r$ in order to construct a representation of $\text{SO}(D)$ on it, e.g. in $D=2$ one has  $\Pi : \text{SO}(2) \mapsto \text{GL}(\sigma_R)$ with $\Pi(2\pi)= \text{id}$ and $\Pi(\alpha)\Pi(\beta)=\Pi(\alpha+\beta)$, where we label the elements of the one-dimensional $SO(2)$ by $\alpha,\beta\in [0,2\pi)$. Without loss of generality we will consider in the following $x_0=0$. Indeed, upon considering a chart in Cartesian coordinates that includes some complete $S_r$ with $r<R/2$ this means we can write the action of a rotation on one of those $S_r$ as a matrix ($x\in S_r$)
\begin{align}\label{rotation_linear}
\Pi(\alpha) \cdot x= \left(\begin{array}{cc}
\cos(\alpha) & \sin(\alpha)\\
-\sin(\alpha) & \cos(\alpha)
\end{array}\right) \cdot x
\end{align}
Note that the rotations for $S_{r\ge R/2}$ are not described by a linear transformation due to the non-trivial boundary conditions. However, as one is ultimately interested in a thermodynamical limit where the infra-red regulator is removed via $R\to \infty$, all rotations of finite distance will have a corresponding $R$ from which on they can be described by (\ref{rotation_linear}). Hence, we choose any $r<R/2$ in the following.

In the remainder of this section we limit the analysis to $D=2$ as, once rotational invariance is established for all rotations in an arbitrary plane, any other rotation can be understood as multiple rotations in suitable planes. Further we employ the ideas of \cite{King86_1,King86_2}: instead of considering arbitrary angles in $[0,2\pi)$, it suffices to show invariance under rotations of only one angle $\theta$ given that $\theta/(2\pi)$ is irrational. This is because the sequence
\begin{equation}
\mathbb{N} \to [0, 2\pi);\; n\mapsto \theta_n:=n\cdot \theta \mod 2\pi
\end{equation}
lies dense in $[0,2\pi)$, i.e.  $\forall \theta' \in [0,2\pi)$ there exists a partial sequence 
$j\mapsto \theta_{n_j}\rightarrow \theta'$. 
Hence we can define the rotation by the angle $\theta'$ as
\begin{equation}\label{ApproximationAngle}
\Pi(\theta'):= \lim_{\theta_{n_j} \rightarrow \theta'} \Pi(\theta)^{n_j}
\end{equation}
It follows, assuming suitable continuity properties, that invariance under all these angles would be established, once it is shown for $\theta$.
In this paper we specialise to the angle $\theta$ defined by $\cos(\theta)=3/5, \sin(\theta)=4/5$ as it is indeed irrational. A proof for this and further properties can be found in \cite{King86_2}.

By the above considerations we can give meaning to the term {\it rotational invariance} as a condition on the continuum Hilbert space measure $\nu$. It is called rotationally invariant 
provided that for any measurable function $g$ we have $\nu(g)=\nu(r(\theta)^\ast\cdot g)$ where 
$(r(\theta)^\ast\cdot g)[\phi]=g[r(\theta)\cdot \phi]$ and $[r(\theta)\cdot \phi](x)=
\phi(\Pi(-\theta)\cdot x)$. Since $\nu$ is defined by its generating functional, we may 
restrict to the functions $g=w[f]$ for which in case of a scalar theory $r(\theta)^\ast w[f]=
w[r(-\theta)\cdot f]$. We now translate this into a condition on the cylindrical projections 
$\nu_M$ of $\nu$ defined by $\nu_M(w_M[f_M]):=\nu(w[I_M f_M])$ where  
\begin{align}
(I_M f_M)(x):=\sum_{m\in\mathbb{Z}^2_M}f_M(m)\chi_{m\epsilon_M}(x)
,\;
\chi_{m\epsilon_M}(x)=\prod_{a=1,2} \chi_{[m^a\epsilon_M,(m^a+1)\epsilon_M)}(x^a)
\end{align}
It follows that $r(-\theta)\cdot I_M f_m$ cannot be written as linear combinations of functions 
of the form $I_M f'_M$ because $r(-\theta)\cdot \chi_{M \epsilon_M}$ is the characteristic 
function of the rotated block. Hence the rotational invariance of $\nu$ does not have a direct 
translation into a condition of the $\nu_M$. While we can define a new embedding map 
by 
\begin{align}
I_{\theta M}: L_{M} &\rightarrow L\\
f_M &\mapsto [I_{\theta M} f_M](x):= \sum_{m\in\mathbb{Z}^2_M}  f_M(m)\chi_{m\epsilon_M}(\Pi(\theta)\cdot x)
\end{align}
the renormalisation flow defined by it may result in a fixed point covariance family
$c^*_{\theta M}$ different from $c^*_M$.
It is therefore a non-trivial question to ask what one actually means by {\it rotational invariance of a discrete lattice theory} or more precisely of a family of corresponding measures.

The idea is to consider both families (i.e. the non-rotated theory described by the covariances $c^*_M$ and the rotated one described by the covariances $c^*_{\theta M}$) as coarse-grained versions of {\it common} finer lattices with spacing $\epsilon_{5M}$ which is why we chose the above particular angle $\theta$.
The rotation of the coarse non-rotated lattice is a sublattice of the fine non-rotated lattice called {\it discrete rotation} and is defined by 
\begin{align}
D_{\theta}:\mathbb{Z}_{M}^2\rightarrow\mathbb{Z}_{5M}^2;\;
m\mapsto \Pi(\theta)\cdot m
\end{align}
This map can be extended to 
\begin{align}
D_\theta \mathbb{Z}_{5M}^2 \rightarrow \mathbb{Z}_{5M}^2
m\mapsto  \lfloor \Pi(\theta) \cdot m \rfloor
\end{align} 
which maps the whole rotated finer lattice into the non-rotated finer lattice. \begin{figure}[h]
\begin{center}
\includegraphics[scale=1]{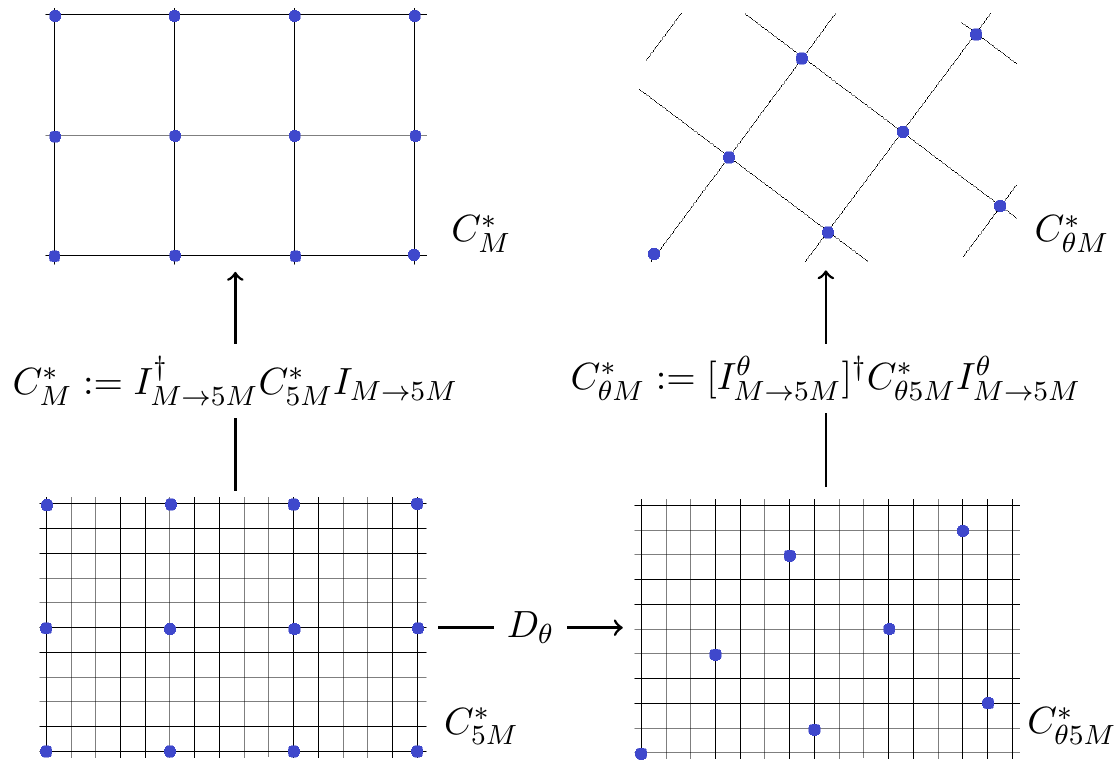}\label{figure1}
\caption{\footnotesize Fixed point covariances $C^*_M,\; C^\ast_{\theta M}$ on lattices rotated relative to each other by the irrational angle $\theta$ (such that $\cos(\theta)=3/5$) can be related by a common refined unrotated lattice and a map $D_{\theta}$, called discrete rotation.}
\end{center}
\end{figure}

The condition that we are about to derive holds for general measures but we also 
note in tandem the corresponding specialisation to Gaussian ones for a later test 
on our model free theory.
Suppose then that $\nu^\ast$ is a rotationally invariant (Gaussian) measure, that is, for its 
generating functional (covariance) we have 
\begin{align}
\nu^{\ast}(w[f])=\nu^\ast(w[\Pi(\theta) \cdot f])\;\;(c^\ast= \Pi(\theta)^\dagger c^\ast \Pi(\theta))
\end{align}
This means that for the cylindrical projections we have the identity
\begin{align}
\nu^\ast_M(w_M[f_M])=\nu^\ast(w[I_M f_M])=\nu^\ast(w[\Pi(\theta) I_M f_M])\;\;
(c^\ast_M=[\Pi(\theta) I_M]^\dagger c^\ast [\Pi(\theta) I_M])
\end{align} 
Now 
\begin{align}
(\Pi(\theta) I_M f_M)(x)=\sum_{m\in \mathbb{Z}_M^2} f_M(m) 
\chi_{m,\epsilon_M}(\Pi(\theta)^{-1}\cdot x)
\end{align}
Let $B_{m,M}$ be the square (block) of which $\chi_{m,\epsilon_M}$ is the characteristic function. 
Then 
\begin{align}
(\Pi(\theta) \cdot \chi_{m, \epsilon_M})(x)=
\chi_{m, \epsilon_M}(\Pi(\theta)^{-1}\cdot x)=\chi_{\Pi(\theta)\cdot B_{m,M}}(x)
\end{align}
is the characteristic 
function of the rotated block of the coarse lattice
with base (lower left corner) now at $\Pi(\theta)\cdot m\in \mathbb{Z}_{5M}^2$. Since we have the disjoint decomposition 
\begin{align}
B_{m,M}=\cup_{m'\in \mathbb{Z}_{5M}^2\cap B_{m,M}}\; B_{m',5M}
\end{align}
we have 
\begin{align} \label{rotationapproximation}
\Pi(\theta) B_{m,M}=\sum_{m'\in \mathbb{Z}_{5M}^2\cap B_{m,M}}\; \Pi(\theta)  B_{m',5M}
\approx 
\sum_{m'\in \mathbb{Z}_{5M}^2\cap B_{m,M}}\; B_{D_\theta \cdot m',5M}
\end{align}
where we have replaced in the last step the rotated blocks of the fine lattice, which before rotation
compose the unrotated block of the coarse lattice, by those unrotated blocks of the fine 
lattice with the bases at the points defined by $D_\theta$. This is an approximation only 
but it is better than one might think because the difference between the two functions 
only affects those blocks $B_{D_\theta \cdot m',5M}$ which intersect the boundary of $B_{\Pi(\theta) \cdot m, M}$. We will come back to the quality of this approximation below. 

In any case, the last line in (\ref{rotationapproximation}) defines an embedding
$I^\theta_{M\to 5M}:\; L_M \to L_{5M}$ by 
\begin{align}\label{4.377}
(I^\theta_{M\to 5M} f_M)(m')=\sum_{m''\in \mathbb{Z}_{5M}^2} \delta_{m',D_\theta \cdot m''}
\sum_{m\in \mathbb{Z}_M^2}\delta_{ m''\in B_{m ,M}} \;f_M(m)
\end{align}
such that $I_{5M} \circ I^\theta_{M\to 5M}$ approximates $\Pi(\theta) \cdot I_M$ in the sense 
specified below. Thus
\begin{align}
\nu^\ast_M(w_M[f_M])\approx \nu^\ast( w[I_{5M} I^\theta_{M\to 5M} f_M])=
\nu^\ast_{5M}(w_{5M}[I^\theta_{M\to 5M}f_M])
\end{align}
or for the Gaussian case
\begin{align}
c^\ast_M\approx [(I_{5M} \circ I^\theta_{M\to 5M}]^\dagger c^\ast [I_{5M}\circ I^\theta_{M\to 5M}]
 =[I^\theta_{M\to 5M}]^\dagger c^\ast_{5M} I^\theta_{M\to 5M}
 \end{align}
To write this just in terms of a single measure (covariance), we use cylindrical consistency\\
$\nu^\ast_M(w_M[f_M])=\nu^\ast_{5M}(w_{5M}[I_{M\to 5M} f_M])$ or 
$c^\ast_N=I_{M\to 5M}^\dagger c^\ast_{5M} I_{M\to 5M}$ to find 
 \begin{align}
\nu^\ast_{5M}(w_{5M}[I^\theta_{M\to 5M} f_M])
\approx
\nu^\ast_{5M}(w_{5M}[I_{M\to 5M} f_M])
\end{align}
or 
\begin{align}
 I_{M\to 5M}^\dagger c^\ast_{5M} I_{M\to 5M}\approx
 (I^\theta_{M\to 5M})^\dagger c^\ast_{5M} I_{M\to 5M}
 \end{align}
as a lattice version for rotational invariance for (Gaussian) measures for scalar field theories.
 
To specify the quality of the approximation depends on the details and properties of the corresponding
measure family. The following result is targeted to the class of Gaussian measures.
\begin{Lemma}
Suppose that $c^\ast$ is the covariance of a rotationally invariant Gaussian measure whose kernel is 
differentiable in the sense of distributions. Then 
\begin{align}   \label{ConditionForRotationalInvariantTheories}
\{c^\ast_M-[I^\theta_{M\to 5M}]^\dagger c^\ast_{5M} I^\theta_{M\to 5M}\}(m_1,m_2)=O(\epsilon_M^5)
\end{align}
for all $m_1,m_2\in \mathbb{Z}_M^2$. The coefficient of $\epsilon_M^5$ is independent of 
$M$. Note that $c^\ast_M(m_1,m_2)=O(\epsilon_M^4)$ in $D=2$.
\end{Lemma}
\begin{proof}
Let $B^\theta_{m,M}=\cup_{m'\in \mathbb{Z}^2_{5M}\cap B_{m,M}} B_{D_\theta m', 5M}$ 
and $S^\theta_{m,M}:=\Pi(\theta) B_{m,M}\cap B^\theta_{m,M}$. Denote  
$\Delta^{\theta +}_{m,M}=\Pi(\theta) B_{m,M}-S^\theta_{m,M}$ and 
$\Delta^{\theta -}_{m,M}=B^\theta_{m,M}-S^\theta_{m,M}$. 
The sets $\Delta^{\theta\pm}_{m,M}$ are homeomorphic since $B^\theta_{m,M}$ consists 
of the squares of $\mathbb{Z}_{5M}^2$ whose lower left corner lies in $\Pi(\theta) B_{m,M}$.
Thus $B^\theta_{m,M}$ lacks parts of $\Pi(\theta) B_{m,M}$ at the left two boundaries 
of $\Pi(\theta) B_{m,M}$ while $B^\theta_{m,M}$ exceeds $\Pi(\theta) B_{m,M}$ at its 
two right boundaries. Hence $\Delta^{\theta\pm}_{m,M}$ are complementary disjoint sets 
whose joint measure is equal to the measure of an integer number of squares of 
the lattice $\mathbb{Z}^2_{5M}$. 
They also have the same Lebesgue measure because $\Pi(\theta) B_{m,M}$ has measure $\epsilon_M^2$ due to rotational invariance of the Lebesgue measure 
and $B^\theta_{m,M}$ has measure 
 $5^2 \epsilon_{5M}^2=\epsilon_M^2$ because $D_\theta$ is injective as is easy to check so that $B^\theta_{m,M}$ consists of 25 squares of the lattice
$\mathbb{Z}^2_{5M}$. Let $h:\Delta^{\theta +}_{m,M}\mapsto \Delta^{\theta -}_{m,M}$ be the 
corresponding homeomorphism which can be written in the form $h(x)=x+g(x) \epsilon_M$ with 
$||g(x)||\le \sqrt{2}$ as the maximal distance between points in the two sets is $\sqrt{2}\epsilon_M$.  
Then by rotational invariance we obtain the third line in
\begin{eqnarray}
&&\{c^\ast_M-[I^\theta_{M\to 5M}]^\dagger  c^\ast_{5M} I^\theta_{M\to 5M}\}(m_1,m_2)=\nonumber\\
&=&\{\int_{ B_{m_1,M}}\;d^2x\int_{B_{m_2,M}}\; d^2y
-\int_{B^\theta_{m1,M}}\;d^2x\int_{B^\theta_{m_2,M}}\; d^2y\}c(x,y)
\nonumber\\
&=&\{\int_{\Pi(\theta) B_{m_1,M}}\;d^2x\int_{\Pi(\theta) B_{m_2,M}}\; d^2y
-\int_{B^\theta_{m1,M}}\;d^2x\int_{B^\theta_{m_2,M}}\; d^2y\}c(x,y)
\nonumber\\
&=&\{\int_{S^\theta_{m_1,M}}\;d^2x\int_{\Delta^{\theta+}_{m_2,M}}\; d^2y
+\int_{\Delta^{\theta+}_{m_1,M}}\;d^2x\int_{S^\theta_{m_2,M}}\; d^2y
+\int_{\Delta^{\theta+}_{m_1,M}}\;d^2x\int_{\Delta^{\theta+}_{m_2,M}}\; d^2y
\nonumber\\
&& -\int_{S^\theta_{m_1,M}}\;d^2x\int_{\Delta^{\theta-}_{m_2,M}}\; d^2y
-\int_{\Delta^{\theta-}_{m_1,M}}\;d^2x\int_{S^\theta_{m_2,M}}\; d^2y
-\int_{\Delta^{\theta-}_{m_1,M}}\;d^2x\int_{\Delta^{\theta-}_{m_2,M}}\; d^2y
\} c(x,y)
\nonumber\\
&=&\int_{S^\theta_{m_1,M}}d^2x\int_{\Delta^{\theta+}_{m_2,M}} d^2y [c(x,y)-c(x,y+g(y) \epsilon_M)]+\nonumber\\
&& +\int_{\Delta^{\theta+}_{m_1,M}}d^2x\int_{S^\theta_{m_2,M}} d^2y[c(x,y)-c(x+g(x)\epsilon_M,y)] \nonumber\\
&& +\int_{\Delta^{\theta+}_{m_1,M}}d^2x\int_{\Delta^{\theta+}_{m_2,M}}d^2y
[c(x,y)-c(x+g(y)\epsilon_M, y+g(y)\epsilon_M]
\end{eqnarray}
from which the claim now follows by considering a power series expansion of $c$.\\
\end{proof}

The lemma does not tell us anything about the size of the coefficient of $\epsilon_M^5$ and 
thus of the actual quality at given $M$, however, assuming that the coefficient is 
finite, for sufficiently large $M$ the approximation error is as small as we want 
compared to the value of the
discrete kernel $c^\ast_M(m_1,m_2)$. 

We translate the approximant 
$c^\ast_{M\theta}:=[I^\theta_{M\to 5M}]^\dagger c^\ast_{5M} I^\theta_{M\to 5M}$ 
whose coefficients are explicitly given by (using translation invariance) 
\begin{align}
c^*_{M\theta}(m)=\frac{1}{5^4}\sum_{\delta_1,\delta_2\in\{0,...,4\}^2}c^*_{5M}(D_{\theta}(5m+(\delta_1-\delta_2)))
\end{align}
into the corresponding Fourier coefficients over which we have better analytic control
\begin{align}\label{prefinalRotInv2}
c^*_{\theta M}(m)&=\sum_{l\in\mathbb{Z}^2_{M}}e^{ik_Ml\cdot m}\hat{c}^*_{\theta M}(l)\\
&=\frac{1}{5^4}\sum_{l\in\mathbb{Z}^2_M}\sum_{\delta \in\{-2,...,2\}^2}e^{ik_Ml\cdot m}e^{ik_MM\delta\cdot m}\sum_{\delta_1,\delta_2\in\{0...4\}^D}e^{ik_{5M}(l+M\delta)\cdot(\delta_1-\delta_2)}\hat{c}^*_{5M}(D_{\theta}(l+M\delta))\nonumber
\end{align}
where we used the fact that $D_{\theta}$ is a bijective map to obtain the third line, as well as for the fourth line, where we have relabelled $D_{\theta}l'\rightarrow l'$ and split $l'=l+M\delta$. We have chosen the interval $\delta\in\{-2,...,2\}^2$ because of its symmetry regarding rotations around the point $x_0=0$, which are considered here, using the periodicity of the boundary conditions.  Performing the sum 
over $\delta_1,\delta_2$ and comparing coefficients we obtain 
\begin{align}\label{finalRotInv}
\hat{c}^*_{\theta M}(l)=\frac{1}{5^4}\sum_{\delta\in\{-2...2\}^2}\prod_{i=1}^2\frac{\sin(\frac{5k_{5M}}{2}[l_i+M\delta_i])^2}{\sin(\frac{k_{5M}}{2}[l_i+M\delta_i])^2}\hat{c}^*_{5M}(D_\theta(l+M\delta))
\end{align}
which can now be readily numerically compared to $c^\ast_M(l)$ (after writing it as an integral
over $k_0$).  

We remark that for rotational invariance under an arbitrary angle $\theta'$ we pick
an approximant $n\cdot \theta$ mod $2\pi$ for sufficiently large $n\in \mathbb{N}$. 
Then, the whole analysis can be repeated using the $M\to 5^n M$ refinement since 
$\Pi(\theta)^n$ is a matrix with rational entries with common denominator $5^n$. 
Since the sets $\Delta^{\theta\pm}_{m, 5^n}$ involve an order of $4\times 5^n$ boundary squares 
of respective measure $\epsilon_{5^n M}^2=\epsilon_M^2 5^{-2n}$ the relative error here would 
even be smaller, i.e. of order $5^{-n} \epsilon_M$.  \\
We leave the detailed numerical analysis of rotational invariance for future publications.\\

\subsection{Example: Numerical investigation for rotational invariance of the free scalar field theory}
In this subsection, we test our criterion numerically using the fixed point theory in $D=2$, which we know to be rotationally invariant in the continuum.\\
First, we verify that the family of covariances $c^\ast_M$ is invariant under rotations by $\pm \pi/2$. 
It suffices to consider the rotation $\Pi(\pi/2)$ and apply this to (\ref{finalIntegralof2DRen}) which is symmetric under exchange of $t_1\leftrightarrow t_2$ (since we could have 
interchanged the roles of those in the contour integral). We have
\begin{align}
\langle r(\pi/2)& f_M, c^*_M r(\pi/2)f'_M\rangle_M=\\
&=\epsilon^4_M\sum_{m,m'\in\mathbb{Z}^2_M}f_M(m)f'_M(m')\sum_{n\in\mathbb{Z}^2_M}e^{ik_M n\cdot (m-m')}\hat{c}^*_M((\Pi(\pi/2)^{-1}\cdot n)_1,(\Pi(\pi/2)^{-1}\cdot n)_2)\nonumber
\end{align}
\noindent Thus, for $\pi/2$ equation (\ref{ConditionForRotationalInvariantTheories}) becomes the condition:
\begin{equation}\label{RIconditionPI/2}
\hat{c}^*_M(n_1,n_2)=\hat{c}^*_M(-n_2,n_1)\; ,\hspace{10pt} \forall n=(n_1,n_2)\in \mathbb{Z}_M^2
\end{equation}
which is fulfilled in case of the free scalar field (\ref{finalIntegralof2DRen}) due to its symmetry and $\cos(-t_i)=\cos(t_i)$.\\

We will now investigate numerically whether the fixed point covariance satisfies the criterion for rotational invariance (\ref{ConditionForRotationalInvariantTheories}). As a sufficient example, we consider the afore-mentioned irrational angle $\theta$, such that $\cos(\theta)=3/5$. Moreover, we will set the IR cut-off to $R=1$ for simplicity and without loss of generality the number of spatial dimensions to  $D=2$. As the value of the mass $p$ and the parameter $k_0$ in (\ref{CovarianceResult2D}) only appear in the combination  $q^2:=(p^2+k_0^2)\epsilon^2_M$, it suffices to fix the latter one to account for both. Here, we choose $q_1^2:=p^2+k_0^2=1$.\\
First, we present the covariance $c_M^*$ itself for $M=40$ in figure \ref{FigureIV_Cov40} where the point $m=(0,0)$ lies in the centre. Due to the periodic boundaries the values on the corners do agree with each other. One can see that the next neighbour interactions drop rapidly with $m\in\mathbb{Z}^2_M=\{0,1,...,M-1\}^2$. The same is true for its Fourier transform $\hat{c}^*_M$. Moreover, the covariance at finite resolution is not invariant under arbitrary rotations, but, heuristically, it appears that the asymmetry could be smoothed out with increasing resolution $M$.\\

\begin{figure}[h]
\begin{center}
\includegraphics[width=0.4\textwidth]{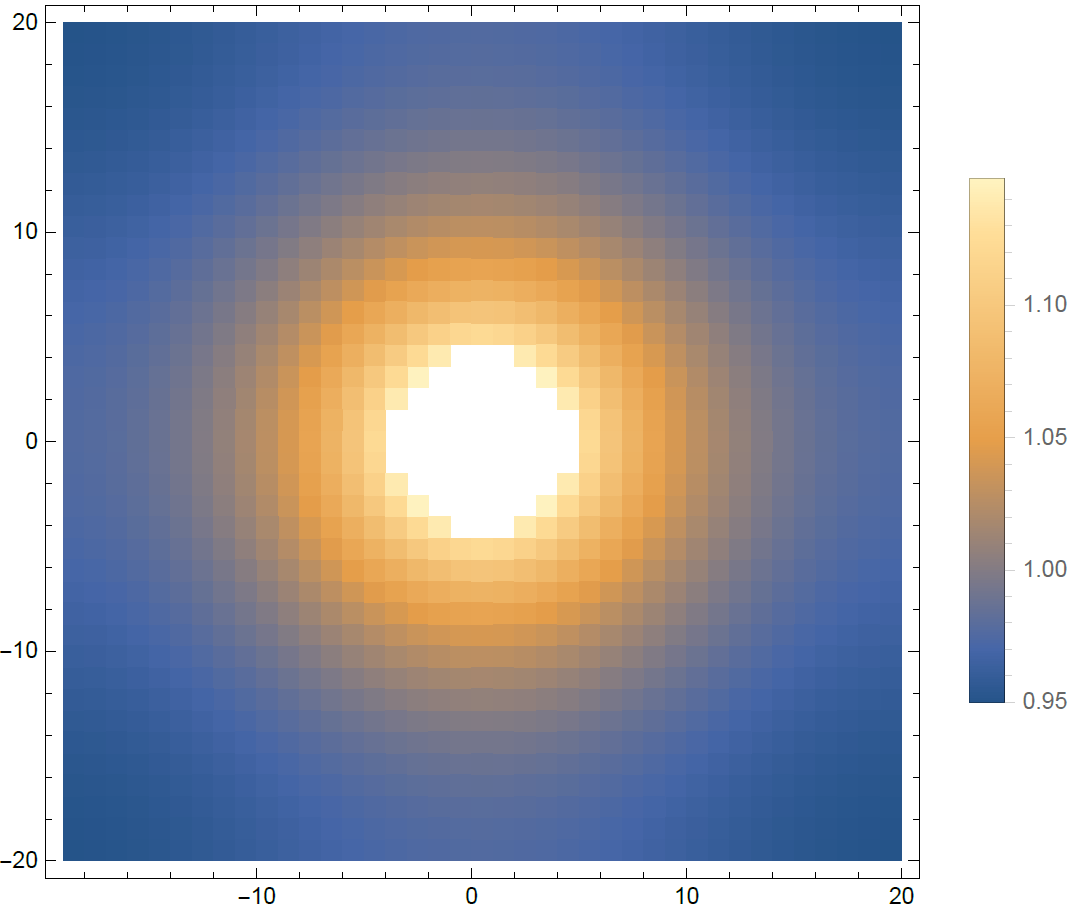}
\includegraphics[width=0.5\textwidth]{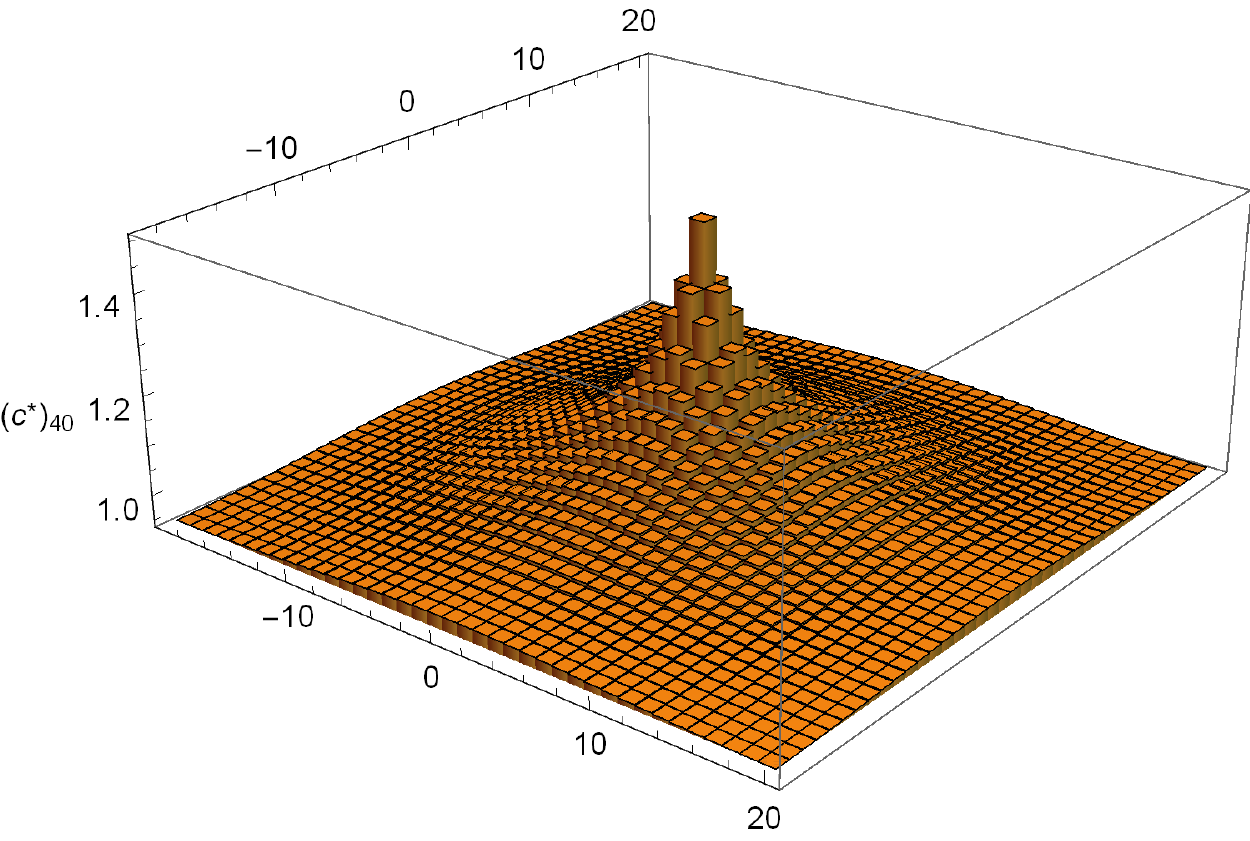}
\caption[Fixed point covariance of resolution $M=40$]{\footnotesize \label{FigureIV_Cov40} 
The covariance $c_M^*(m)$ of the fixed point theory in $D=2$ spatial dimensions. We have chosen the IR cut-off $R=1$, mass $p=1$, and $k_0=0$. The torus $[0,1)^2$ is approximated by a lattice with $M=40$ points in each direction, where the point $m=(0,0)$ lies in the centre of the plotted grid. As one can see, the contributions from next neighbour frequencies are highly suppressed.}
\end{center}
\end{figure}
\begin{figure}[h]
\begin{center}
\includegraphics[width=0.44\textwidth]{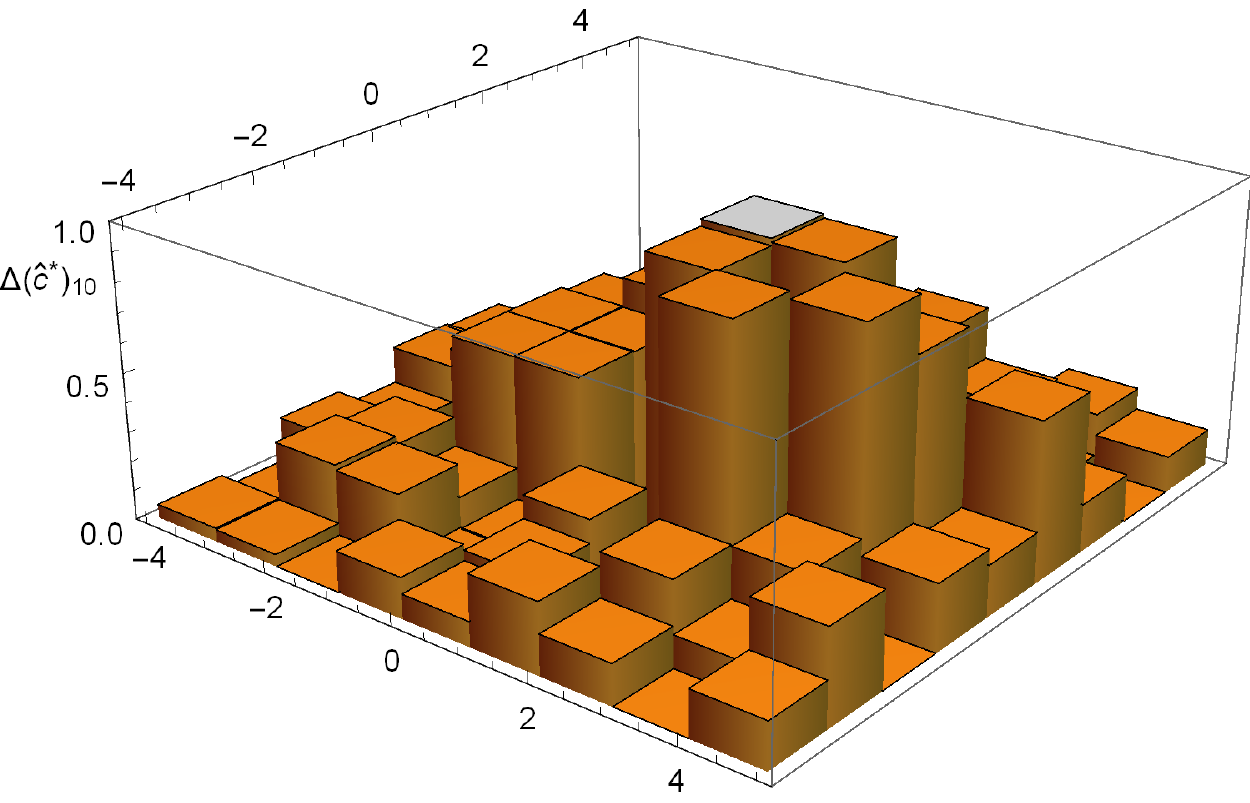}
\includegraphics[width=0.44\textwidth]{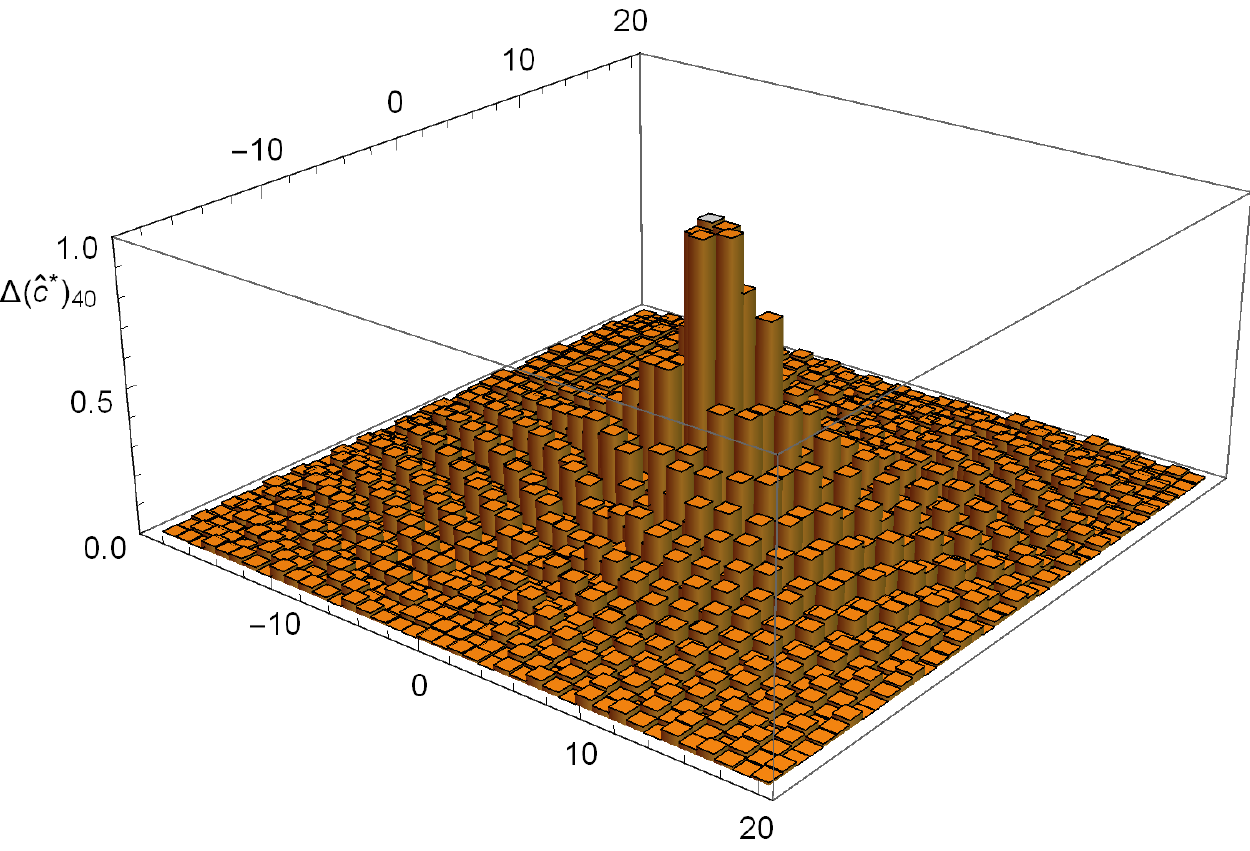}
\includegraphics[width=0.44\textwidth]{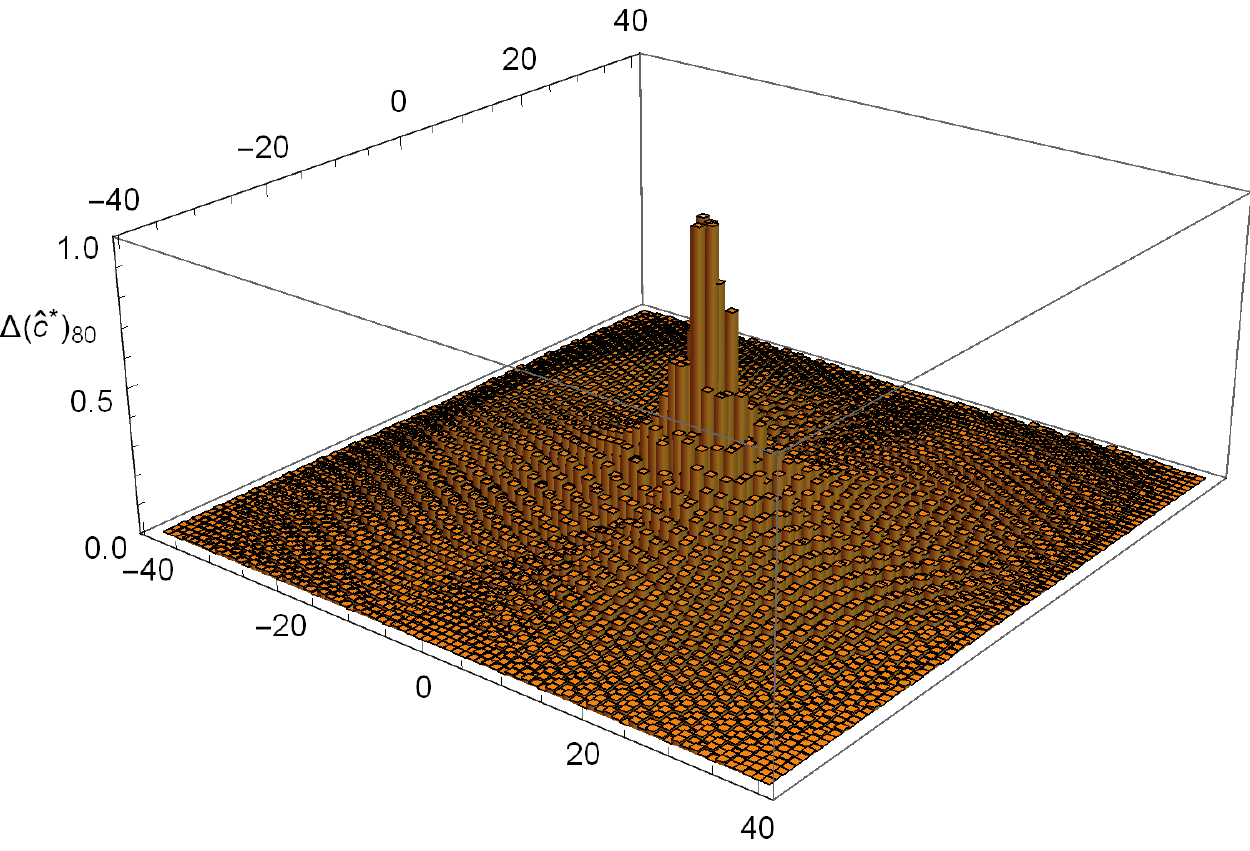}
\includegraphics[width=0.44\textwidth]{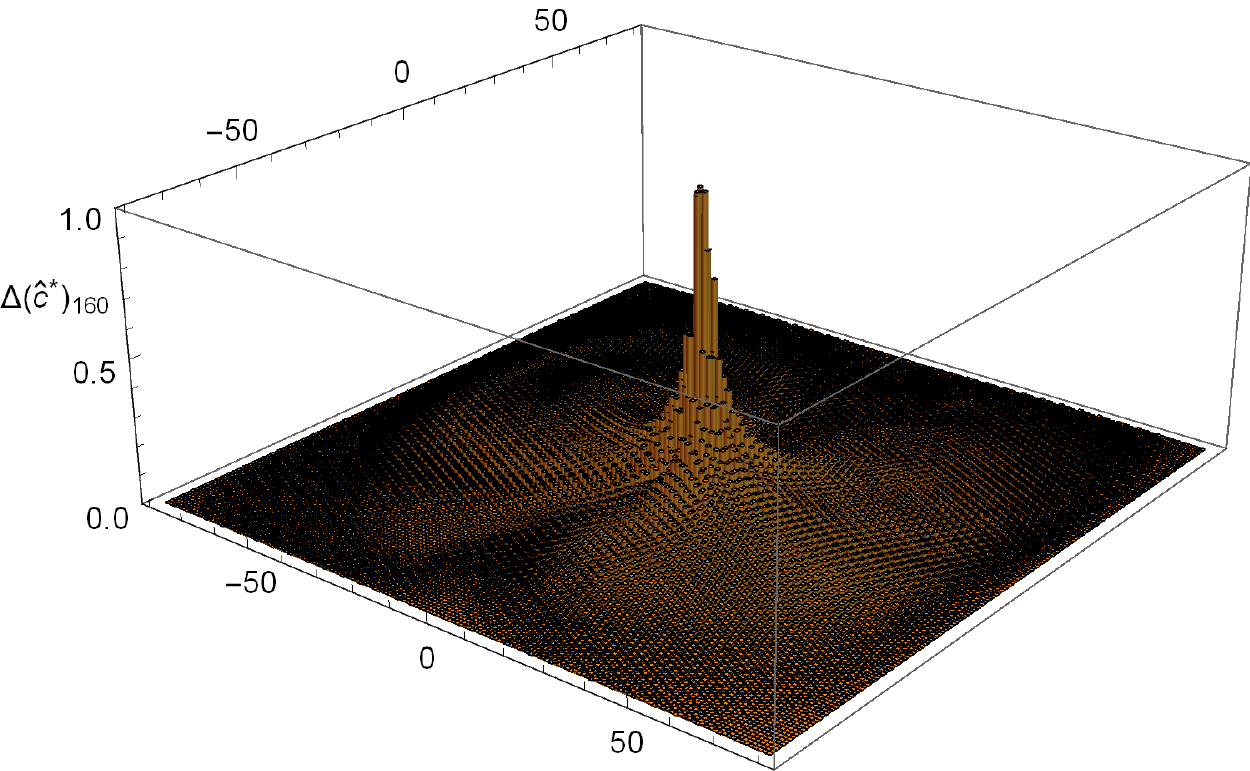}
\caption[Comparison of rotational invariance at different resolutions]{\footnotesize \label{FigureIV_Error} 
For lattices of size $M=10,\; 40,\; 80,\;160$ the relative deviation $\Delta\hat{c}_M^*(l)=[|\hat{c}_M^*-\hat{c}^*_{\theta,M}|/\hat{c}^*_M](k_0=0,l)$ is plotted for $l\in \mathbb{Z}_{M}^2$ with mass $p=1$ and IR cut-off $R=1$ . High values of $\Delta\hat{c}_M^*$ indicate non-invariance of the covariance at given resolution under rotations. (The {\it grey} data point lies outside the plotted range of $[0,1)$, with numerical value $\approx 40$.)  We find that the relative deviation is non-vanishing everywhere at finite resolution, however it decreases with $M^{-1}$ because $\hat{c}^*_M-\hat{c}^*_{\theta M} \sim O(\epsilon_M^5) \sim \hat{c}^*_M \epsilon_M$. This is the approximative behaviour of a rotationally invariant fixed point theory. For $M=160$ the computed covariance features already rotational invariance to a high precision.}
\end{center}
\end{figure}

Next, we consider the quality of the approximant to the rotated covariance as $M$ varies. This approximant, $c^*_{M\theta}$, is the Fourier transform of (\ref{finalRotInv}) and should agree with the unrotated covariance $c^*_M$ up to a mistake $\mathcal{O}(\epsilon_M^5)$, given the fixed point covariance restores rotational invariance in the continuum. As the same must be true for their Fourier transforms, we consider $\hat{c}^*_{M\theta}$ and $\hat{c}^*_M$ on lattices of different size $M$ and study whether their deviation decays appropriately. Both covariances are of order $\mathcal{O}(\epsilon^4_M)$, hence their {\it relative deviation} should decay with $\mathcal{O}(\epsilon_M)$:
\begin{align}
\Delta \hat{c}^*_M(l):=\frac{|\hat{c}^*_M(0,l)-\hat{c}^*_{M\theta}(0,l)|}{\hat{c}^*_M(0,l)}\sim \mathcal{O}(\epsilon_M)
\end{align}
That it decays indeed fast, is shown in figure \ref{FigureIV_Error} for lattices of size $M=10,40,80$ and $160$. Although at low resolution the covariance features a high discrepancy with the approximant $\hat{c}_{M\theta }^*$, the relative deviation $\Delta\hat{c}^*_M$ becomes smaller as the resolution of the spatial manifold increases. Only in a neighbourhood of the centre of the grid, i.e. the point around which we rotate, the approximation fails. But, this neighbourhood shrinks linearly with the resolution $M$. For $M=160$ the computed covariance already features rotational invariance to a high precision.\\
To study the decay behaviour of $\Delta\hat{c}^*_{M}(l)$ further, one could now consider the characteristic function $\chi_B$ of a region $B\subset[0,1)^2$ and compare, for different resolutions $M$, the mean $\overline{\Delta\hat{c}_M}[\chi_B]$ of the relative deviation in this region, i.e. the mean of $\Delta\hat{c}^*_{M}(l)$ over all $l\in\mathbb{Z}^2_M$ such that $\text{supp}(\chi_l)\subset B$. For example, for $l_0=(0,2)\in\mathbb{Z}_5^2$ let the support of $\chi_{l_0}$ be the region of interest, i.e. on resolution $M_0=5$ we have $\overline{\Delta\hat{c}_{M_0}}[\chi_{l_0}]=\Delta\hat{c}^*_{5}(l_0)$. At any other $M\in5\mathbb{N}$, we consider the refinement in $\mathbb{Z}_M^2$, i.e. the points $\frac{M}{5} l_0+\delta$ for $\delta\in [0,M/5-1]^2$. We find that the mean 
\begin{align}
\overline{\Delta\hat{c}_M}[\chi_{l_0}]:=\frac{1}{(M/5)^2}\sum_{\delta\in [0,M/5-1]^2}\Delta\hat{c}^*_{M}(\frac{M}{5} l_0+\delta)
\end{align}
is decaying with $M^{-1}$, see figure \ref{FigureIV_Decay} for two examples. It confirms that (\ref{ConditionForRotationalInvariantTheories}), i.e. the condition for rotational invariance, is satisfied up to an error of $\epsilon_M=M^{-1}$ to a very high precision for the considered examples and thus indicates that rotational invariance will be recovered in the continuum.\\

\begin{figure}[H]
\begin{center}
\includegraphics[width=0.44\textwidth]{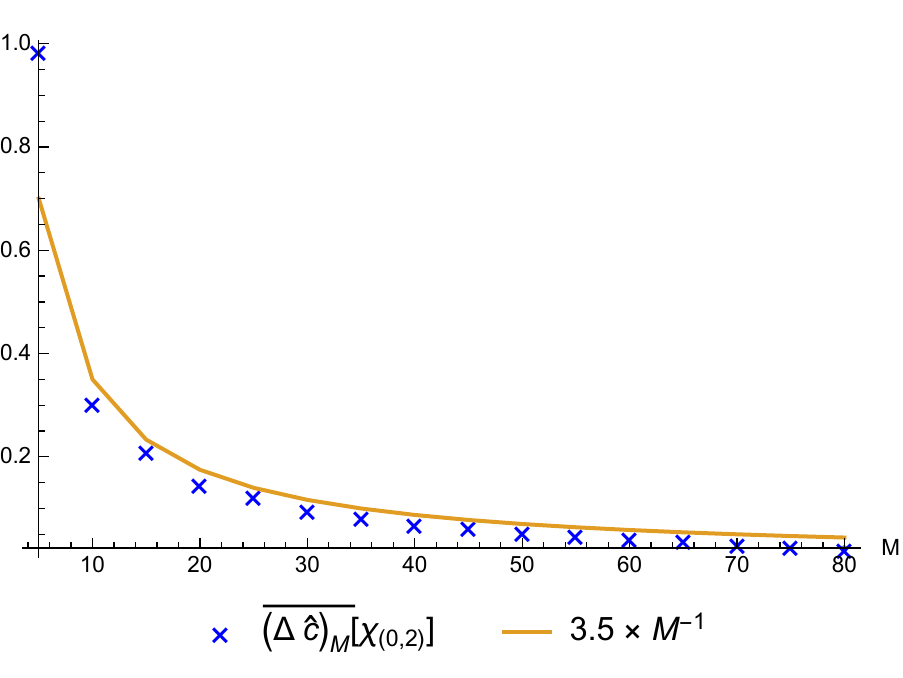}
\includegraphics[width=0.44\textwidth]{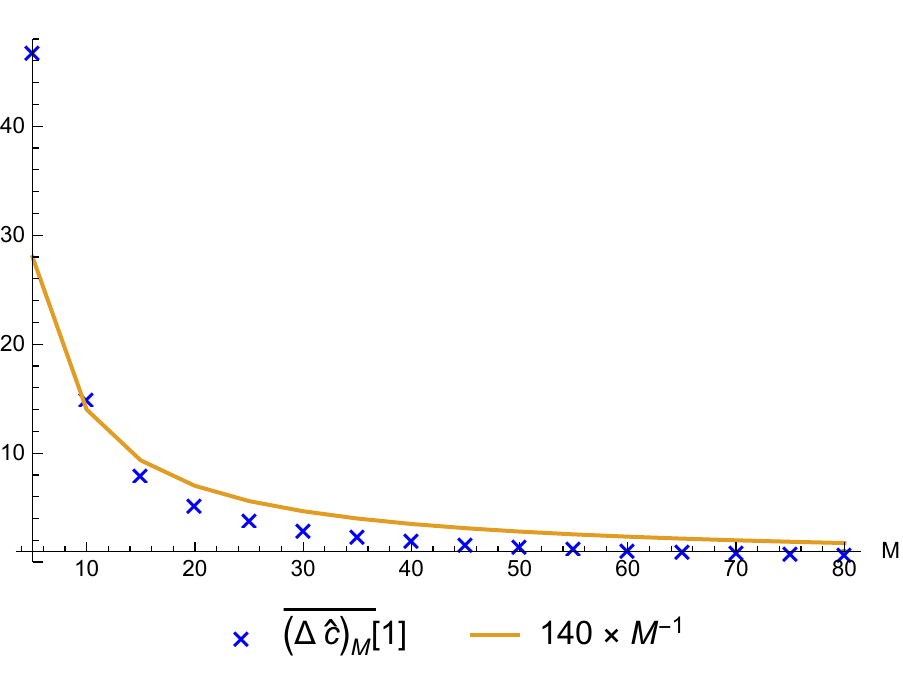}
\caption[Decay behaviour of the relative deviation from rotational invariance for increasing resolution]{\footnotesize \label{FigureIV_Decay} 
The decay behaviour of the mean $\overline{\Delta\hat{c}_M}[\chi]$ of the relative deviation over a region with characteristic function $\chi$ is presented. For two distinct regions, we compute it at different resolutions $M$. On the left, $\chi_{(0,2)}$ is the characteristic function of this block that can be associated with the point $m_0=(0,2)$ on resolution $M=5$. The values for $\overline{\Delta\hat{c}_M}[\chi]$ are shown in blue and we approximate the decay behaviour by the function $f(M)=3.5\;M^{-1}$ (in orange). Similarly, $\chi=1$ is associated with the whole torus $[0,1)^2$ and is presented on the right. Here, the decay is best approximated by $f(M)=140\times M^{-1}$ in orange. These two are arbitrary cases, however we expect the decay to be of this form for each region. This confirms that the decay behaviour is sufficiently fast to account for a rotationally invariant fixed point theory.}
\end{center}
\end{figure}

\section{Conclusion}
\label{Conclusion}

While the framework of covariant Renormalization \cite{Riv91,Comb04,Bal89a,Bal89b,Dim12a,Dim12b,Dim12c} has a long story of success, its implementation and application in the Hamiltonian setting have remained largely unexplored. With this series of papers, we have made first steps in this direction: In \cite{LLT1} 
we motivated a Hamiltonian renormalisation scheme and in \cite{LLT2,LLT3} we tested this scheme 
successfully for the free massive Klein Gordon field in $D+1=2$ spacetime dimensions.
In this paper we extended this test successfully to arbitrary dimensions. This extension 
made it possible also to test the robustness of the fixed point under changes of the 
coarse graining map which defines the renormalisation flow as well as how to test the 
rotational invariance of the fixed point theory by using the finite resolution projections of the 
corresponding Hilbert space measure. The latter test is useful in situations in which the  computation of the 
continuum measure is too complicated, but in which one has at least numerical access to its
approximate finite resolution cylindrical projections obtained by iterating the flow equations 
sufficiently often, provided its convergence or at least the existence of a fixed point 
can be established.  

The next step in our programme will be to extend the framework to gauge theories, as the most interesting models of modern physics are phrased in this language, e.g. QCD, and test whether the validity of the direct Hamiltonian renormalisation is also given therein. Afterwards, one could apply the framework in the context of gravity. This includes the ``Asymptotic Safety'' programme \cite{RS02,Per09,RS12} as well as other approaches to quantum gravity, such as Loop Quantum Gravity (LQG) \cite{Rov04,AL04,Thi07}. As the latter one was originally formulated in the canonical setting, it is hoped that the strategy outlined in this series could fixed the quantization ambiguities arising when defining the constraint 
or Hamiltonian operators, see e.g. \cite{LT16,Thi96_1,Thi96_2,ALM15}. In the context 
of LQG different regularisation schemes (e.g. based on different ordering prescriptions)  lead to different operators \cite{DL17a,DL17b} and the goal must be to find ideally unique constraint operators of general relativity which do not display those ambiguities anymore. A strict 
criterion is to obtain a theory free of anomalies for the symmetries of the theory.
\section*{Acknowledgements}
Part of this work was financially supported by a grant from the Friedrich-Alexander University to the Emerging Fields Project ``Quantum Geometry" under its Emerging Fields Initiative. 
TL thanks the Heinrich-B\"oll Foundation for financial support.
KL thanks the German National Merit Foundation for financial support.\\
\\

\begin{appendix}

\section{Detailed computations of section \ref{Hamiltonian renormalisation}}
\label{sa}

In this appendix we fill the gaps which were left out in section  \ref{Hamiltonian renormalisation} of the main text. The first paragraph investigates the fact that the multi-dimensional renormalisation transformation (\ref{Covarianceflow}) decouples in its spatial directions. The second paragraph describes how the integral in (\ref{finalIntegralof2DRen}) can be performed.\\

As discussed in the main section, the initial covariance can be factorised into two factors which very closely resemble the 1+1 dimensional case
\begin{equation} \label{startingdecoupledcovariance_app}
\hat{c}^{(0)}_M(l)=-\oint_{\gamma} dz \frac{\epsilon_M^4}{8\pi i}\; 
\frac{1}{q_1^2(z)/2+1-\cos(t_1)}\;\;\frac{1}{q_2^2(z)/2+1-\cos(t_2)}
\end{equation}
Let us now focus on the precise action of the map (\ref{Covarianceflow}), by writing it in terms of its kernel $c^{(n)}_M(m_1',m_2')=c^{(n)}_M(m'_1-m'_2)$:
\begin{equation}
C^{(n+1)}_M (m'_1-m'_2)= 2^{-2D}\sum_{\delta',\delta''\in\{0,1\}^D}C^{(n)}_{2M}(2m'_1+\delta'-2m'_2+\delta'')
\end{equation}
and correspondingly for the Fourier transform for $D=2$
\begin{align} \label{flowdefinition}
&\hat{c}^{(n+1)}_M(l)=2^{-4}\sum_{\delta,\delta',\delta''\in\{0,1\}^2}\hat{c}^{(n)}_{2M}(l+\delta M) e^{ik_{2M}(l+\delta M)\cdot(\delta'-\delta'')}=\nonumber\\
&=\frac{1}{2^4} \sum_{\delta_1,\delta_2\in\{0,1\}}\hat{c}^{(n)}_{2M}(l_1+\delta_1 M, l_2 +\delta_2 M)\left(
e^{ik_{2M}(l_1+l_2+(\delta_1+\delta_2)M)}+\right.\nonumber\\
&\hspace{20pt} \left.+e^{-ik_{2M}(l_1+l_2+(\delta_1+\delta_2)M)}+e^{ik_{2M}(l_1-l_2+(\delta_1-\delta_2)M)}+e^{-ik_{2M}(l_1-l_2+(\delta_1-\delta_2)M)}\right.\nonumber\\
&\hspace{20pt} \left.+2e^{ik_{2M}(l_2+\delta_2M)}+2e^{-ik_{2M}(l_2+\delta_2M)}+2e^{ik_{2M}(l_1+\delta_1M)}+2e^{-ik_{2M}(l_1+\delta_1M)}+4\right)
\end{align}
where we wrote explicitly all 16 terms stemming from the different combinations of $(\delta'-\delta''$).
\begin{align} \label{decoupling}
&=\frac{1}{2^4} \sum_{\delta_1,\delta_2=0,1}\hat{c}^{(n)}_{2M}(l_1+\delta_1M,l_2+\delta_2M)\left( 4+ 4\cos(k_{2M}(l_2+\delta_2M))+4\cos(k_{2M}(l_1+\delta_1M))\right.+\nonumber\\
&\hspace{20pt} \left. +2\cos(k_{2M}(l_1+\delta_1M)+k_{2M}(l_2+\delta_2M))+2\cos(k_{2M}(l_1+\delta_1M)-k_{2M}(l_2+\delta_2M))\right)\nonumber\\
&=\frac{1}{2^2}\sum_{\delta_1,\delta_2=0,1}\hat{c}^{(n)}_{2M}(l_1+\delta_1M,l_2+\delta_2M)\times\nonumber\\
&\hspace{20pt} \left( 1+ \cos(k_{2M}(l_2+\delta_2M))+\cos(k_{2M}(l_1+\delta_1M))+\cos(k_{2M}(l_1+\delta_1M))\cos(k_{2M}(l_2+\delta_2M))
\right)\nonumber\\
&=\frac{1}{4}\sum_{\delta_1,\delta_2=0,1}\left(1+\cos(k_{2M}(l_1+\delta_1 M)\right)\left(1+\cos(k_{2M}(l_2+\delta_2 M)\right)\hat{c}^{(n)}_{2M}(l_1+\delta_1M,l_2+\delta_2M)
\end{align}
where we have used in the second step, that $2\cos(x)\cos(y)=\cos(x+y)+\cos(x-y)$. One realises that both directions completely decouple in the renormalisation 
transformation. Since the initial covariance factorises  under the contour integral over $\gamma$
this factorisation is preserved under the flow and implies that the flow of the covariance in each direction can be performed separately. \\

Following the arguments in section \ref{Hamiltonian renormalisation} one can determine the fixed point covariance stemming from (\ref{startingdecoupledcovariance_app}) for each direction separately and finds with $t_j=k_M l_j$ 
\begin{align}\label{finalIntegralof2DRen}
\hat{c}^*_M(k_0,l)&=-\left(\frac{\epsilon_M^4}{2\pi i}\right)\; \oint_\gamma\; dz\;
\prod_{j=1,2} \frac{1}{q_j^3}
\frac{q_j \text{ch}(q_j) - \text{sh}(q_j) +(\text{sh}(q_j)-q_j)\cos(t_j)}{\text{ch}(q_j)-\cos(t_j)}
\end{align}
Note that it is not necessary to pick a square root of the complex parameter 
$q^2_{1,2}(z)=\epsilon_M^2(\frac{k_0^2+p^2}{2}\mp z)$ since the integrand only depends 
on the square, despite its appearance (in other words, one may pick the branch arbitrarily,
the integrand does not depend on it). It follows that the integrand is a single valued 
function of $z$
which is holomorphic everywhere except for simple poles which we now determine, and which 
allow to compute the contour integral over $\gamma$ using the residue theorem.

There are no poles at $q_{1,2}^2=0$ since the functions $[q \text{ch}(q)-\text{sh}(q)]/q^3,\;
[\text{sh}(q)-q]/q^3$ are regular at $q=0$. Hence the only poles come from the 
zeroes of the function $\text{ch}(q)-\cos(t)$. Using $\text{ch}(iz)=\cos(z)$ and the periodicity of 
the cosine function we find $iq=\pm[t+2\pi N]$ with $N\in \mathbb{Z}$ or 
$q^2=-(t+2\pi N)^2$. In terms of $q_{j},\;j=1,2$ this means that 
\begin{align}
(k_0^2+p^2)/2\mp z=-\frac{(t_j+2\pi N)^2}{\epsilon_M^2}\;\;\Leftrightarrow\;\;
z=z_{N}=\pm [(k_0^2+p^2)/2+\frac{(t_j+2\pi N)^2}{\epsilon_M^2}]
\end{align}
It follows that the second factor involving $q_2$ has no poles in the domain bounded by $\gamma$
because they all lie on the negative real axis while those coming from the factor involving
$q_1$ lie all on the positive real axis. We will denote the latter by $z_N$. The 
poles coming from the zeroes of $\text{ch}(q_1)-\cos(t_1)$ are simple ones as one can check 
by expanding the hyperbolic cosine at $z_N$ in terms of $z-z_N$, in other words
\begin{align}
\lim_{z\to z_N}\frac{z-z_N}{\text{ch}(q_1(z))-\cos(t_1)}=
\lim_{z\to z_N}\frac{1}{\text{sh}(q_1(z)) q_1'(z)}=
\lim_{z\to z_N}\frac{2 q_1(z)}{\text{sh}(q_1(z)) [q_1^2(z)]'}=
-\frac{2q_1(z_N)}{\epsilon_M^2{\rm sh}(q_1(z_N))}
\end{align}
which is again independent of the choice of square root. We have used de l$^\prime$ Hospital's 
theorem in the second step. Note that $q_1(z_N)^2=-(t_1+2\pi N)^2$ and 
$q_2(z_N)^2=q^2+(t_1+2\pi N)^2:=q_N^2$ where $q^2:=\epsilon_M^2(k_0^2+p^2)$. 
 
Performing the integral using the residue theorem and using that 
$\text{ch}(q_1(z_N))=\cos(t_1)$ and $\text{sh}(q_1(z_N))=\pm i\sin(t_1)$ (sign cancels against 
similar choice for $q_1(z_N)$) we end up with 
\begin{align}\label{CovarianceResult2D}
\hat{c}_M^*(k_0,l)=& \epsilon^2_M\frac{[q_N {\rm ch}(q_N)-{\rm sh}(q_N)]+[{\rm sh}(q_N)-q_N]\cos(t_2)}{q_N^3[{\rm ch}(q_N)-\cos(t_2)]}+\nonumber\\
&-2\epsilon^2_M\underset{N\in\mathbb{Z}}{\sum}\frac{\cos(t_1)-1}{(2\pi N +t_1)^2}
\frac{1}{q_N^3} \;\frac{q_N \text{ch}(q_N)-\text{sh}(q_N)+(\text{sh}(q_N)-q_N)\cos(t_2)}
{\text{ch}(q_N)-\cos(t_2)}
\end{align}

\section{Derivation of the fixed point for RG-flow of prime $u=3$}
\label{sb}

The following explicit calculations are performed for the prime  $u=3$ as this illustrates what needs to be done also in the general case. The initial data of the RG-flow is given for $D=1$ with $t=k_Ml, q^2=\epsilon_M^2(k_0^2+p^2)$ by
\begin{equation}\label{3start_app}
\hat{c}^{(0)}_M(k_0,l)=\frac{\epsilon^2_M}{2(1-\cos(t))+q^2}
\end{equation}
In order to compute this flow, it is useful to recall the trigonometric addition theorems for 
the cosine function
\begin{equation}
\cos(x)+\cos(y)=2\cos\left(\frac{x+y}{2}\right)\cos\left(\frac{x-y}{2}\right),\hspace{10pt}
\cos(x)\cos(y)=\frac{1}{2}\left(\cos(x-y)+\cos(x+y)\right)
\end{equation}
to note the following explicit values
\begin{equation}
\cos(\frac{1}{6}2\pi)=\frac{1}{2}, \hspace{20pt} \cos(\frac{1}{3}2\pi)=-\frac{1}{2}, \hspace{20pt}\cos(\frac{2}{3}2\pi)=-\frac{1}{2}
\end{equation}
and to employ the  {\it Chebyshev recursive method}, which states that for $N\in\mathbb{N}$:
\begin{equation}
\cos(Nx)=2\cos(x)\cos((N-1)x)-\cos((N-2)x)
\end{equation}
which is an easy expansion into exponentials and finds application in what follows 
for the case $N=3$ and $x \rightarrow x/3$ to express $\cos(x)=2 \cos(x/3) \cos(2/3x)-\cos(x/3)$.

Equipped with these tools, we start to compute the RG-flow of $I_{M\rightarrow 3M}$ by finding a common denominator of the sum in (\ref{generaldecoupling}) assuming $\hat{c}^{(n)}$ could have been written in the form
\begin{equation} \label{3start2}
\hat{c}^{(n)}_M(k_0,l)=\frac{\epsilon^2_M}{q^3}\frac{b_n(q)+c_n(q)\cos(t)}{a_n(q)-\cos(t)}
\end{equation}
with suitably chosen functions $a_n,b_n,c_n$ of $q$ as we already know is true for (\ref{3start}). Then, the common denominator after one renormalisation step is generated 
by the linear combination of the of the three fractions in (\ref{generaldecoupling}) 
and is given by:
\begin{align*}
&\left[a_n(q)-\cos\left(\frac{t}{3}\right)\right]\left[a_n(q)-\left(\frac{t}{3}+M\frac{2\pi}{3M}\right)\right]\left[ a_n(q)-\cos\left(\frac{t}{3}+\frac{2}{3}2\pi\right)\right]=\\
&=a_n(q)^3-a_n(q)^2\left[\cos\left(\frac{t}{3}\right)+\cos\left(\frac{t}{3}+\frac{1}{3}2\pi\right)+\cos\left(\frac{t}{3}+\frac{2}{3}2\pi\right)\right]_A\\
&\hspace{15pt}+a_n(q)\left[\cos\left(\frac{t}{3}\right)\cos\left(\frac{t}{3}+\frac{1}{3}2\pi\right)+\cos\left(\frac{t}{3}\right)\cos\left(\frac{t}{3}+\frac{2}{3}2\pi\right)+\cos\left(\frac{t}{3}+\frac{1}{3}2\pi\right)\cos\left(\frac{t}{3}+\frac{2}{3}2\pi\right)\right]_B\\
&\hspace{15pt}-\left[\cos\left(\frac{t}{3}\right)\cos\left(\frac{t}{3}+\frac{1}{3}2\pi\right)\cos\left(\frac{t}{3}+\frac{2}{3}2\pi\right)\right]_C
\end{align*}
Each of the three prefactors in front of each power of $a_n(q)$ can now be evaluated precisely with the methods stated above. We obtain:
\begin{align}
\left[\cos\left(\frac{t}{3}\right)\right.&\left.+\cos\left(\frac{t}{3}+\frac{1}{3}2\pi\right)+\cos\left(\frac{t}{3}+\frac{2}{3}2\pi\right)\right]_A=0\\
\left[\cos\left(\frac{t}{3}\right)\right.&\left.\cos\left(\frac{t}{3}+\frac{1}{3}2\pi\right)+\cos\left(\frac{t}{3}\right)\cos\left(\frac{t}{3}+\frac{2}{3}2\pi\right)+\cos\left(\frac{t}{3}+\frac{1}{3}2\pi\right)\cos\left(\frac{t}{3}+\frac{2}{3}2\pi\right)\right]_B=-\frac{3}{4}\\
\left[\cos\left(\frac{t}{3}\right)\right.&\left. \cos\left(\frac{t}{3}+\frac{1}{3}2\pi\right)\cos\left(\frac{t}{3}+\frac{2}{3}2\pi\right)\right]_C=\frac{1}{4}\cos(t)
\end{align}
So we get for the denominator
\begin{equation}\label{fixpointeq3}
\frac{1}{4}\left(\left[4a_n(q)^3-3a_n(q)\right]-\cos(t)\right)
\end{equation}
which is again of the form that (\ref{3start2}) had. Moreover, we note that the $t$-independent part of (\ref{fixpointeq3}) is exactly the right hand side of the triple-angle formula for cos, cosh:
\begin{equation}
a(3q)=4a(q)^3-3a(q)
\end{equation}
hence with the choice of $a(q)=\text{ch}(q)$ we have found a fixed point for the flow induced onto the $a_n(q)$.\\
For the numerator, we continue in the same manner. After some pages of calculation, one finds 
it to be given by
\begin{align}
&(3+4\cos(t/3)+2\cos(2/3t))(b_n-c_n\cdot \cos(t/3))(a_n-\cos(t/3+2\pi/3))\times\nonumber\\
&\times(a_n-\cos(t/3+2/3\cdot2\pi))(3+4\cos(t/3+2\pi/3)+2\cos(2t/3+2/3\cdot 2\pi))(a_n-\cos(t/3))\times\nonumber\\
&\times(a_n-\cos(t/3))(a_n-\cos(t/3+2/3\cdot 2\pi))(3+4\cos(t/3+2/3\cdot 2\pi)+2\cos(2/3t+4/3 \cdot 2\pi))
\times\nonumber\\
&\times(b_n-c_n\cdot \cos(t/3+2/3 \cdot 2\pi))(a_n-\cos(t/3+2\pi/3)(b_n-c_n \cdot \cos(t/3+2\pi/3)))\nonumber\\
\label{3nominatorflow}
&=\ldots=
\left(-\frac{3}{4}+6a_n+9a_n^2\right)b_n-6a_n(1+a_n)c_n+\frac{3}{4}\left(4(1+a_n)b_n-(3+4a_n+4a_n^2)c_n\right)\cos(t)
\end{align}
Thus, also the numerator is cast again into an expression of the form $b_{n+1}+c_{n+1}\cos(t)$. We can already make use of the fact, that at the fixed point one has $a=\cosh(q)$. Making an educated guess and trying whether
\begin{equation}\label{educatedguess}
b=q \text{ch}(q)-\text{sh}(q),\hspace{30pt} c=\text{sh}(q)-q
\end{equation}
are solutions of the fixed point equation determined by (\ref{3nominatorflow}) one uses the triple-angle formula for the sine function
\begin{equation}
\sin(3x)=2\cos(x)\sin(2x)-\sin(x)=-2\cos(2x)\sin(x)-\sin(x)
\end{equation}
and obtains indeed by plugging (\ref{educatedguess}) into (\ref{3nominatorflow}):
\begin{align}
3(1+\cosh(q))(q \text{ch}(q))-\frac{1}{4}(3+4\text{ch}(q)+4\text{ch}(q)^2)(-\text{sh}(q)+q)=\frac{3}{4}\left[-3q+\text{sh}(3q)\right]
\end{align}
and
\begin{align}
\left(-\frac{3}{4}+6\text{ch}(q)+9\text{ch}(q)^2\right)( \text{ch}(q)-\text{sh}(q))-6\text{ch}(q)(1+\text{ch}(q))(-\text{sh}(q)+q)=\frac{3}{4}\left[3q \text{ch}(3q)-\text{sh}(3q)\right]
\end{align}
which presents indeed the triple angle formula, up to the common prefactor of $3/4$. The 
factor $1/4$
gets cancelled by the pre factor of $1/4$ in front of $a_n$ in (\ref{fixpointeq3}).  The factor $3$ 
cancels against a factor $3^{-1}$ which is obtained as follows: the map itself was defined with a prefactor $3^{-2}$, the factor $q^{-3}$ gives $3^3$ and the $\epsilon_M^2$ gives $3^{-2}$ which altogether gives a factor $3^{-1}$. Hence we have indeed found exactly the same fixed point 
under the $M\rightarrow3M$ coarse graining map as we found for the $M\to 2M$ 
coarse graining map!

\end{appendix}

}

\end{document}